\documentclass[11pt]{article}
\usepackage{gbm}
\usepackage{setspace}
\usepackage{cprotect}
\usepackage{makeidx}
\usepackage[mathscr]{euscript}
\usepackage{accents}

\def\anon{0}
\def\authornotes{0}

\ifnum\authornotes=1
    \newcommand{\tselil}[1]{{\color{blue} [T: #1]}}
    \newcommand{\eric}[1]{{\color{red} [E: #1]}}
	\newcommand{\todo}[1]{{\color{Red}[TODO: #1]}}	
\else
    \newcommand{\tselil}[1]{}
    \newcommand{\eric}[1]{}
	\newcommand{\todo}[1]{}
\fi

\title{Polynomial-time sampling despite disorder chaos}
\ifnum\anon=0
\author{Eric Ma\thanks{Stanford University. \texttt{eryma@stanford.edu}.} \and Tselil Schramm\thanks{Stanford University.  \texttt{tselil@stanford.edu}. Supported by NSF CAREER award \# 2143246.}}
\else
\author{Author name(s) withheld for double-blind review}
\fi

\date{}

\begin{document}
\maketitle

\begin{abstract}
A distribution over instances of a sampling problem is said to exhibit \emph{transport disorder chaos} if perturbing the instance by a small amount of random noise dramatically changes the stationary distribution (in Wasserstein distance).
Seeking to provide evidence that some sampling tasks are hard on average, a recent line of work has demonstrated that disorder chaos is sufficient to rule out ``stable'' sampling algorithms, such as gradient methods and some diffusion processes.

We demonstrate that disorder chaos does not preclude polynomial-time sampling by canonical algorithms in canonical models.
We show that with high probability over a random graph $\bG \sim G(n,1/2)$: (1) the hardcore model (at fugacity $\lambda = 1$) on $\bG$ exhibits disorder chaos, and (2) Glauber dynamics run for $O(n)$ time can approximately sample from the hardcore model on $\bG$ (in Wasserstein distance).

\end{abstract}
\thispagestyle{empty}

\clearpage

{ \hypersetup{hidelinks} \tableofcontents }
\thispagestyle{empty}
\clearpage
\setcounter{page}{1}

\section{Introduction}

% Sampling is at least as hard as optimization
% Just one example: the hardcore model
    % Define the hardcore model
    % Can take fugacity larger to concentrate on large independent sets
    % Worst-case hardness comes from NP-completeness
    % Sampling is also believed to be hard on average for large enough fugacity, in, say, G(n,p)
% Average-case hardness for other problems: spin glasses, etc
% Sometimes sampling is believed to be hard even though optimization is easy
    % SK model

Suppose we are given an instance of an optimization problem, such as the maximum independent set problem on a graph $G$.
Instead of solving the optimization problem and returning a maximum independent set $S^*$ in $G$, we may instead wish to \emph{sample} from a distribution $\mu_G$ over solutions, where the probability of sampling an independent set $S$ is a function of the objective value (here, size) of $S$.
For example, the well-studied \emph{hardcore model} with \emph{fugacity} $\lambda > 0$ assigns to each set $S \subset V(G)$ a probability 
\[
\mu_G(S) \propto \lambda^{|S|} \cdot \Ind[S \text{ independent set in } G].
\]
If $\lambda = 1$, $\mu_G$ is the uniform distribution over independent sets in $G$; on the other hand, as $\lambda \to \infty$, $\mu_G$ becomes concentrated on large independent sets, eventually converging to the uniform measure over optimal solutions $S^*$.
This sampling problem and its cousins originate at the intersection of theoretical computer science and statistical physics, and sampling is by now mainstream in theoretical computer science.

The discussion above makes it clear that optimization reduces to sampling.
For example, in the hardcore model, sampling from $\mu_G$ at sufficiently large fugacity $\lambda$ gives a randomized algorithm for solving the maximum independent set problem on $G$.
For this reason, if $\mathrm{P} \neq \mathrm{NP}$, we expect that sampling in general requires super-polynomial time.
Even in the average case, when the graph $\bG$ is drawn from, say, the \erdos-\renyi distribution, the best polynomial-time algorithms known today can only find an independent set of size $\frac{1}{2}$ of optimal, and so even in average-case graphs, sampling from $\mu_{\bG}$ at large enough fugacity is (conjecturally) hard.

In fact, for some optimization problems it is believed that sampling is strictly harder than optimization, in that the sampling problem appears to be hard, even when the optimization problem is known to be solvable in polynomial time.
%Consider, for example, the maximum cut problem on a weighted graph (also called an \emph{Ising model}).
%For a choice of \emph{inverse temperature} $\beta > 0$, the graph with weighted adjacency matrix $A$ induces a distribution $\mu_{A}$ over vertex cuts $x \in \{\pm 1\}^n$, with $\mu_{A}(x) \propto \exp(-\beta x^\top A x)$.
%When $\bA \in \R^{n \times n}$ is symmetric with independent $\N(0,\frac{1}{n})$, this is known as the \emph{Sherrington-Kirkpatrick} model.
%A striking result of Montanari \cite{} demonstrated a polynomial-time algorithm for approximately optimizing the quadratic form $x^\top \bA x$.
%On the other hand, when $\beta > 1$, the \emph{sampling} problem is not known to be solvable in polynomial time, and is conjectured hard.
%\tselil{Should I say something about the fact that there is a phase transition in $\mu_{\bA}$? Or too much going on already}
One well-studied example is the max-cut problem on a graph with independent standard Gaussian weights, also known as the \emph{Sherrington Kirkpatrick Ising model}.
There, a striking result of Montanari \cite{Mon21} gives a PTAS for the optimization problem, but the associated sampling problem is conjectured hard (see e.g. \cite{AMS22}).
The same phenomenon occurs in several average-case sampling/optimization problems, including polynomial optimization with Gaussian coefficients (``$p$-spin models'') \cite{Subag21,AMS23,HMP24}, and finding a Boolean vector in an intersection of random halfspaces (``perceptron problems'') \cite{BS20,ALS22,AG24}.
It is possible (though not widely believed, see discussion below) that this also occurs in the hardcore model in sparse random graphs $\bG \sim G(n,d/n)$, in that polynomial-time optimization algorithms produce independent sets of size $(1+o_d(1))\frac{\log d}{d} n$, but we know no efficient algorithm that samples from $\mu_{\bG}$ at the corresponding fugacity $\lambda = \Omega( 1/\log d)$.

% How to substantiate the average-case hardness?
% Recent work in p-spin models suggests ruling out restricted classes of algorithms with properties of the measure
% Define the disorder chaos property
% Define stable algorithm
% Explain the triangle inequality argument

We want to understand whether average-case sampling problems, such as sampling from the hardcore model on \erdos-\renyi graphs, are hard for polynomial-time algorithms.
Complexity cannot always follow from the optimization-to-sampling reduction, as sometimes the associated optimization problem is easy.
So far the leading approach is to prove lower bounds against specific algorithms.
For example, we might try to show super-polynomial mixing time of the canonical \emph{Glauber dynamics} Markov Chain, whose stationary distribution is $\mu_{\bG}$. 
Such lower bounds are a good place to start, but they are dissatisfying in that they only rule out a very specific (albeit canonical) algorithm.

Recently, El Alaoui, Montanari, and Sellke \cite{AMS22} have proposed \emph{transport disorder chaos} as a more general criterion for hardness of sampling, and have shown that it gives unconditional lower bounds against \emph{smooth} algorithms.
A sampling problem $\mu_{\bG}$ exhibits transport disorder chaos if small random perturbations $\bG'$ of the instance $\bG$ tend to result in a large Wasserstein (or ``transport'') distance between $\mu_{\bG}$ and $\mu_{\bG'}$
(this is qualitatively similar to, though technically weaker than, the overlap-gap property, see the survey \cite{Gamarnik21}).
Smooth algorithms are algorithms which are robust to small perturbations: $\calA$ is a smooth sampling algorithm if the Wasserstein distance between the output distributions $\calA(\bG)$ and $\calA(\bG')$ is small.
Thus, given that $\mu_{\bG}$ exhibits transport disorder chaos, the triangle inequality implies that smooth algorithms cannot sample from $\mu_{\bG}$ in Wasserstein distance.
%Though they studied the Sherrington-Kirkpatrick Ising model, for the sake of illustration, we here consider again the hardcore model over \erdos-\renyi graphs.
%Let $\bG \sim G(n,p)$, and let $\bG'$ be another graph given by resampling each edge of $\bG$ independently with some probability $s \in (0,1)$.
%A sampling algorithm $\calA$ is said to be \emph{stable} if with high probability over $\bG,\bG'$, the distributions $\calA(\bG)$ and $\calA(\bG')$ are $o_s(1)$-close in Wasserstein distance; that is, if one can couple $\bS\sim\calA(\bG)$ and $\bS' \sim \calA(\bG')$ so that $|\bS\cap\bS'| = (1-o_s(1))|\bS|$.
%The measure $\mu_{\bG}$ is said to exhibit \emph{transport disorder chaos} if $\mu_{\bG}$ and $\mu_{\bG'}$ are $\Omega_s(1)$-far in Wasserstein distance; that is, a small amount of noise makes it impossible to couple $\mu_{\bG}$ and $\mu_{\bG'}$ in such a way that the independent sets coincide on $1-o_s(1)$ vertices.
%El Alaoui, Montanari, and Sellke observe that if $\mu_{\bG}$ satisfies transport disorder chaos, then by the triangle inequality for the Wasserstein distance, no smooth algorithm $\calA$ can produce $\calA(\bG)$ which is close to $\mu_{\bG}$ in Wasserstein distance; otherwise
%\[
%\Omega_s(1) = \dist(\mu_{\bG},\mu_{\bG'}) \le \dist(mu_{\bG},\calA(\bG)) + \dist(\calA(\bG),\calA(\bG')) + \dist(\calA(\bG'),\mu_{\bG'}) = o_s(1).
%\]

El Alaoui, Montanari, and Sellke introduced this idea in the context of the Sherrington Kirkpatrick Ising model as evidence that their algorithmic method is close to tight (at least among smooth algorithms): they present a smooth sampling algorithm based on stochastic localization that works up to \emph{inverse temperature} (the analogue of fugacity) $\beta = \frac{1}{2}$ (later improved to $\beta <1$ in \cite{Cel24}), and then show that disorder chaos onsets at larger inverse temperature $\beta > 1$.
This style of lower bound has since been utilized for other problems, including $p$-spin models \cite{AMS23,AMS25}, perceptron problems \cite{AG24}, and Ising models on sparse graphs \cite{HPP25}.

% Which sampling algorithms are stable?
% O(1) steps of gradient descent, stochastic localization, some markov chains
% These are the algorithms known to work in the easy regime.
Which natural sampling algorithms are smooth?
In \cite{AMS22,AMS23}, the authors prove that in $p$-spin models, an algorithm comprising of $O(1)$-steps of gradient descent or discretized stochastic localization is stable.
In the context of $p$-spin models, these algorithms are state-of-the-art for sampling (and the disorder chaos lower bounds of \cite{AMS22,AMS23} conclusively show that their analysis of these algorithms is sharp).
But for other average-case sampling problems, it is not clear whether the canonical algorithms are smooth.
Further, even in $p$-spin models, gradient descent is not smooth if it is run for $\omega(1)$ steps\footnote{The top eigenvector of a symmetric Gaussian matrix is computable by gradient descent, but is not a smooth function of the matrix.} (though for most $p$-spin models, analyzing $\omega(1)$ steps seems beyond the reach of current techniques). 
% \eric{Here, do you mean you don't know how to show it is smooth, or that it is not smooth? I vaguely sense that we may have discussed this recently} \tselil{I know how to show it is not smooth; it follows from the fact that gradient descent can find the top eigenvector of a Gaussian matrix, and the top eigenvector exhibits chaos}

In a recent work, Li and Schramm \cite{LS24} begin to investigate whether disorder chaos could be a barrier to a larger class of algorithms.
They demonstrate that the (not especially canonical) short $(s,t)$-path problem in $G(n,2\log n / n)$ on the one hand exhibits transport disorder chaos, but on the other hand admits polynomial-time sampling algorithms.
However, their sampling algorithms are based on brute-force enumeration: the corresponding measure $\mu_{\bG}$ is concentrated on $\poly(n)$ efficiently-enumerable solutions.
Most sampling problems of interest have $\mu$ anti-concentrated on any subset of $\poly(n)$ solutions.
Thus, sampling by brute-force enumeration rarely yields efficient algorithms. 
Adding to this that $(s,t)$-path has not been previously studied as a sampling problem, and Li and Schramm's counterexample might seem like a curious anomaly, without clear implications for the complexity of typical sampling problems.

% The main question we are trying to address is: is this a satisfying type of lower bound? 
% Prior work of Li and Schramm samples despite disorder chaos via brute-force enumeration, but the problem is unique in that it is enumerable in polynomial time, which is a special feature of that model.
% What about running these canonical algorithms for $\omega(1)$ steps?
% Can other canonical algorithms sample in the presence of disorder chaos?
The main question we wish to address in this work is whether transport disorder chaos implies hardness of sampling by canonical algorithms.
We ask whether one can furnish a more canonical sampling problem which is a counterexample in that it (1) can be solved by polynomial-time algorithms, while (2) exhibiting transport disorder chaos.
Indeed, we are able to show that at fugacity $\lambda = 1$, polynomial-time Glauber dynamics solves the problem of sampling from the hardcore model over $\bG \sim G(n,1/2)$ (in Wasserstein distance), even while $\mu_{\bG}$ exhibits disorder chaos.
Our sampling result may be of independent interest, as no efficient algorithm is known to sample from the hardcore model on $G(n,1/2)$ at this fugacity.

% We study the canonical hardcore model on G(n,1/2) and show that despite the presence of disorder chaos, at fugacity $\lambda = 1$ Glauber dynamics samples in polynomial time.

\subsection{Results}

% To state our results, we require a definition:
% \begin{definition}[Normalized Wasserstein distance]
% If $\mu,\nu$ are distributions over a space equipped with metric $\|\cdot\|$, the \emph{normalized $2$-Wasserstein distance between $\mu$ and $\nu$} is the quantity
% \[
% \w_2(\mu,\nu) \defeq \inf_{\pi \in C(\mu,\nu)} \sqrt{\frac{\E_{(\bX,\bY) \sim \pi} \|\bX-\bY\|^2}{\sqrt{\E_{\bX \sim \mu}\|\bX\|^2 }\sqrt{\E_{\bY \sim \nu} \|\bY\|^2}}},
% \]
% where $C(\mu,\nu)$ is the space of all couplings of $\mu$ and $\nu$.
% \end{definition}
% \tselil{I actually think a definition that might make more sense is to separately normalize $\bX$ by its expected norm and $\bY$ by its expected norm. If I am not mistaken this shouldn't affect anything in our proofs... also since we have the definition here, we can get rid of it in prelims}

To state our results, we require a definition:
\begin{definition}[Normalized Wasserstein/transport distance]
\label{def:w2}
If $\mu,\nu$ are distributions over a space equipped with $\ell_2$ metric $\|\cdot\|$, the \emph{normalized $2$-Wasserstein or transport distance between $\mu$ and $\nu$} is the quantity
\[
\w_2(\mu,\nu) \defeq \inf_{\pi \in C(\mu,\nu)} \sqrt{ \E_{(\bX,\bY) \sim \pi} \left\|\ \frac{\bX}{\sqrt{\E_{\bX \sim \mu}\|\bX\|^2 }}-\frac{\bY}{\sqrt{\E_{\bY \sim \nu} \|\bY\|^2}}\right\|^2 },
\]
where $C(\mu,\nu)$ is the space of all couplings of $\mu$ and $\nu$.
\end{definition}

We will represent $\mu_G$ as a distribution over the $\ell_2$-space $\set{0,1}^n$, associating each independent set $S$ with its indicator vector.
Our first result is that the hardcore model on $G(n,1/2)$ exhibits disorder chaos, meaning that small perturbations of $\bG$ result in large perturbations of $\mu_{\bG}$, as measured in normalized Wasserstein distance:

% State theorem about disorder chaos
% Explain the proof in one line
\begin{theorem}[The hardcore model exhibits disorder chaos]
\label{thm:chaos-intro}
Consider correlated graphs $\bG,\bG' \sim G(n,1/2)$, where $\bG'$ is given by re-sampling each edge in $\bG$ independently with probability $s \in (0,1)$.
Let $\mu_{\bG},\mu_{\bG'}$ denote the hardcore model at fugacity $\lambda = 1$ on $\bG,\bG'$ respectively.
Then on average, even for arbitrarily small $s$ the distributions $\mu_{\bG}$ and $\mu_{\bG'}$ are as separated as possible in normalized Wasserstein distance,
\[
\lim_{s \to 0} \liminf_{n \to \infty} \E_{\bG,\bG'} \w_2(\mu_{\bG},\mu_{\bG'})^2 = 2.
\]
\end{theorem}
The intuition for \Cref{thm:chaos-intro} is straightforward.
Since we are representing subsets of $[n]$ with their indicators in $\{0,1\}^n$, the squared distance between independent sets $\bS \in \bG$ and $\bS' \in \bG'$ is $\|\bS-\bS'\|^2 = |\bS| + |\bS'| - 2 |\bS \cap \bS'|$.
The quantity $|\bS \cap \bS'|$ is at most the size of the largest subset of vertices in $\bS$ which form an independent set in $\bG'$.
Say $|\bS| = m$; since each edge in $\bG'$ is resampled with probability $s$, the graph induced on $\bS$ in $\bG'$ is distributed according to $G(m, s/2)$.
The key fact is that in $G(m,s/2)$, the size of the maximum independent set scales as $2 \log m / \log (1-s/2)^{-1}$ with high probability.
Hence, so long as $s=\Omega_m(1)$,  $ |\bS \cap \bS'| = O(\log m) \ll m= |\bS|$, and $\|\bS - \bS'\|^2 = |\bS| + |\bS'| - o(|\bS|)$.
At fugacity $\lambda = 1$, the sizes of $\bS \sim \mu_{\bG}$ and $\bS' \sim \mu_{\bG'}$ concentrate well and grow as $n \to \infty$, which suffices to argue that $\w_2^2(\mu_{\bG},\mu_{\bG'}) = 2(1-o_n(1))$ with high probability.

We remark that the same argument can be used to prove that $\w_2(\mu_{\bG},\mu_{\bG'}) = \Omega_s(1)$ when $\bG \sim G(n,p)$ for any $\frac{1}{2} > p = \Omega(1/n)$ and $\lambda = \Theta_n(1)$ (though one must make adjustments to the asymptotics of $|\bS \cap \bS'|$ when $p$ is small).
For the sake of simplicity, we have fixed $p = \frac{1}{2}$ and fugacity $\lambda = 1$.

\medskip

We also show that despite the disorder chaos, the canonical Markov Chain sampling algorithm for the hardcore model induces a distribution which is close to $\mu_{\bG}$ in $\w_2$ distance.

\begin{theorem}[Polynomial-time sampling in Wasserstein distance]
\label{thm:glauber-sample}
Suppose $\bG \sim G(n,1/2)$ and let $\mu_{\bG}$ be the hardcore model on $\bG$ at fugacity $\lambda = 1$.
For an integer $k \ge 0$, let $\mu^{\mathrm{Glauber}}_k(\bG)$ be the distribution of the Glauber dynamics Markov Chain on ${\bG}$ after initializing from the empty set, and running until the stopping time $\btau = \min\{100 (\log n)2^k), \btau_k\}$,  
where $\btau_k$ is the first time the Markov Chain hits an independent set of size $k$. 
Then for $k = \log n - 30\log\log n$, with probability $1-o_n(1)$,
\[
\w_2(\mu^{\bG},\mu^{\mathrm{Glauber}}_{k} (\bG)) = o_n(1).
\]
\end{theorem}

\begin{remark}
For simplicity our theorem is stated and proved only at fugacity $\lambda=1$, but similar arguments should work for any $\lambda \leq (\log n)^{O(1)}$. 
\end{remark}

The Glauber Dynamics is a Markov Chain, whose state at time $t$ is an independent set $\bS_t$ of $\bG$.
At step $t+1$, the dynamics samples a vertex $v \sim [n]$ uniformly at random.
If $v \in \bS_t$, the dynamics decides whether to drop $v$ from $\bS_t$ based on the outcome of a biased coinflip; otherwise if $v \not\in\bS_t$ but $\{v\}\cup\bS_t$ forms an independent set, $v$ is added to $\bS_t$ based on the outcome of a biased coinflip.

To prove \Cref{thm:glauber-sample}, we first give a simple argument that couples the trajectory of Glauber Dynamics initialized at $\bS_0 = \emptyset$ up to time $\btau_k$ with the (randomized) greedy algorithm for finding an independent set in $\bG$.
Concentration phenomena for nested independent sets in $\bG$ allow us to show that if the greedy algorithm is run for $k = \log n - \Omega(\log\log n)$ steps, the output is close in total variation distance to the uniform distribution over $k$-independent-sets in $\bG$.
We construct a good $\w_2$ coupling by (1) sampling a $k$-clique $\bS$ using Glauber dynamics, (2) sampling an integer $k^+$ according to the marginal over $|\bS'|$ for $\bS' \sim \mu_{\bG}$, and (3) if possible, outputting a uniformly random $k^+$-sized clique $\bS'$ which contains $\bS$ (otherwise output an arbitrary $k^+$-independent set in $\bG$).
We argue that $k^+$ is concentrated in a window of size $2\log\log n$ around $\log n$, so that with high probability $k^+ > k$ and step (3) is valid.
% \eric{Here, we use $k$ to denote the cliques we actually output, and $k^+$ to denote the cliques we can couple to. However, later in the proof, we use $k^-$ and $k$, respectively. Do you think this is an issue? I was going to swap the notation here, but above, you use $k$ as the size that Glauber outputs, so I'm not entirely sure what to do.}
Then, we establish that for each $k^+$, the distribution induced by first sampling a $k$-clique, and then sampling a random $k^+$ clique which contains it, is close to uniform over $k^+$ cliques in $\bG$.
Our proof is not too complicated: the conclusion follows in a straightforward way from sufficiently strong concentration of the number of independent sets of size $k^+$ containing $S$ in $\bG$; we establish concentration by bounding $\Omega(\log\log n)$ moments via combinatorial arguments.
Since $\mu_{\bG}(S)$ treats all independent sets of equal size symmetrically, and since $|\bS\cap \bS'|/|\bS'| = 1-o_n(1)$, the coupling witnesses a $\w_2$ distance of $o_n(1)$.

% State theorem about Glauber dynamics
% Briefly explain proof

\subsection{Discussion and open problems}
% Discussion: should disorder chaos be taken as a sign of hardness, in this model or in others?

Together, \Cref{thm:chaos-intro,thm:glauber-sample} demonstrate that even for very canonical sampling problems, disorder chaos is not incompatible with the success of very canonical polynomial-time (Wasserstein) sampling algorithms.
The upshot is that even when run on average-case sampling problems, many canonical sampling algorithms are simply not smooth in the sense of El Alaoui, Montanari, and Sellke.
Our results demonstrate that the set of non-smooth algorithms includes Glauber dynamics in the hardcore model at some fugacities; presumably this is true of Glauber dynamics in other contexts as well.\footnote{We mention that Glauber dynamics is known \cite{BJ18} to mix slowly for many $p$-spin models at low temperatures, though as far as we know there is no lower bound for the Sherrington Kirkpatrick model or other models with ``full replica symmetry breaking.''}

A more troubling point is that the class of non-smooth algorithms includes iterative algorithms, such as gradient methods or stochastic localization, when run for $\omega(1)$ steps. 
The state-of-the-art algorithms for sampling from $p$-spin models are such iterative methods run for $O(1)$ steps, and disorder chaos has been used as some evidence of their optimality.
Though it is not really our expectation that $\poly(n)$ steps of a stochastic localization procedure can sample from $p$-spin models in the hard regime, our results still sow some doubt on the finality of the disorder-chaos based lower bounds.

% Context of computational landscape: 
    % Problem well-studied in G(n,d/n), where sampling in TV is easy at fugacity < uniqueness and potentially hard above (Eric: can sample above uniqueness threshold in https://arxiv.org/abs/2504.03406 on random regular graphs. Maybe we should also very briefly contrast the proofs; these results sampling at small fugacity generally have something to do with correlation decay/spectral independence).
    % The problem here is incomparable since those works think of d fixed, lambda fixed 
    % Open problems: sampling in TV? Sampling in W2 in G(n,d/n)?

A separate question concerns sampling from the hardcore model.
The best-studied average-case sampling question for the hardcore model concerns \emph{sparse} random graphs $\bG \sim G(n,d/n)$ when $d=O_n(1)$, and specifically the question of sampling in \emph{total variation}: the goal is to have a polynomial-time algorithm $\calA$ which with high probability satisfies $\dtv{\mu_{\bG}}{\calA(\bG)} = o(1)$.
By analogy with the $d$-regular tree, it has been conjectured that Glauber Dynamics can sample in polynomial time in this setting for fugacity up to the reconstruction threshold at $\lambda = \Omega_d(1)$ (see \cite{RSVVY14,BST16}).
The largest fugacity at which Glauber Dynamics is currently known to sample from $\mu_{\bG}$ in total variation is $\lambda = O(d^{-1})$ \cite{BGGS24,EF23}; in the qualitatively similar random $d$-regular graph model, this is improved to $\lambda = O(d^{-1/2})$ \cite{CCCYZ25} (and the recent result \cite{LMRW24} for the Ising models on $G(n,d/n)$ might generalize to allow $\lambda = O(d^{-1/2})$ in the hardcore model).
It is difficult to make an apples-to-apples comparison between fugacity in $G(n,1/2)$ and $G(n,d/n)$, so instead, we compare the typical size of an independent set $\bS \sim \mu_{\bG}$: when $\bG \sim G(n,1/2)$, our \Cref{thm:glauber-sample} shows that the Glauber dynamics can sample in Wasserstein when $\lambda$ is set such that the sets $\bS \sim \mu_{\bG}$ have size within a factor $1/2$ of optimal.
On the other hand, when $\bG \sim G(n,d/n)$ or from the random $d$-regular graph model, the largest fugacities at which Glauber is known to sample in total variation produce sets that are only within a factor up to $1/4$ of optimal.

Given our \Cref{thm:glauber-sample}, in $G(n,1/2)$ Glauber can sample in Wasserstein distance in a regime where it is not currently known to sample in total variation distance. 
Of course, it may be the case that Glauber can sample in total variation up to $\lambda = \Omega(1)$ in both dense and \emph{sparse} \erdos-\renyi graphs. 
Our proof strategy cannot be directly used to establish such a result because of a lack of concentration; one would have to use different arguments.
But alternatively, it is also possible that the conjecture derived from the $d$-regular tree is incorrect, and that sampling in Wasserstein distance is strictly easier than sampling in total variation.
We are intrigued to know the answer, and we wonder what evidence for this possibility could look like.

\subsection*{Organization}
In \Cref{sec:prelims} we discuss notation and give a few definitions and standard technical lemmas. 
\Cref{sec:hardcore} establishes concentration properties of the hardcore model on $G(n,1/2)$ (specifically, that the sizes of independent sets drawn from the hardcore model are very well concentrated around a specific value $\ks$).
In \Cref{sec:disorder}, we prove that the hardcore model on \erdos-\renyi graphs exhibits disorder chaos (\Cref{thm:chaos-intro}), and in \Cref{sec:glauber} we prove that Glauber Dynamics samples from the hardcore model in transport distance (\Cref{thm:glauber-sample}).
\section{Preliminaries}
\label{sec:prelims}

In this section we describe our notation conventions, give some definitions, and mention some useful standard technical lemmas.

We will follow the convention that random variables are written in bold font.
We will occasionally abbreviate independent set as i.s.
In this work, $\log\defeq\log_2$ unless otherwise specified.
We use standard big-Oh notation; at times we use the possibly-ambiguous $O(1)$; where there is room for confusion, we add a subscript to signify the variable with respect to which the quantity is bounded, say $f(m) = O_m(1)$ to signify that $\lim_{m \to \infty} f(m) < \infty$ or $f(\eps) = O_\eps(1)$ to signify that $\lim_{\eps \to 0} f(\eps) < \infty$.
For a measure $\mu$ on an $\ell_2$ space, we will let $m_2(\mu) = \sqrt{\E_{\bX\sim\mu} \norm{\bX}_2^2}$.
For probability measures $\mu,\nu$, let $C(\mu,\nu)$ denote the set of all couplings of $\mu,\nu$.

We will find the following standard results useful.
\begin{lemma}
\label{W2_convexity}
Let $\mu = \sum_{i=1}^k p_i \mu_i, \nu = \sum_{i=1}^k p_i \nu_i$ be mixture distributions. Consider any set of couplings $\pi_i,i\in[k]$ where $\pi_i \in C(\mu_i,\nu_i)$ is a coupling of $\mu_i$ and $\nu_i$. Then
\begin{align*}
\w_2(\mu,\nu)^2 \leq \sum_{i=1}^k p_i \E_{(\bX,\bY) \sim \pi_i}\left\|\ \frac{\bX}{\sqrt{\E_{\bX \sim \mu}\|\bX\|^2 }}-\frac{\bY}{\sqrt{\E_{\bY \sim \nu} \|\bY\|^2}}\right\|^2
\end{align*}
\end{lemma}
\begin{proof}
Observe that the following is a coupling of $\mu,\nu$: first pick $\bi\in[k]$ according to the distribution $p_1,\dots,p_k$, and then sample $(\bX,\bY)\sim\pi_{\bi}$. The conclusion follows.
\end{proof}

We also mention that our definition of $\w_2$ satisfies the triangle inequality. This follows from the fact that the standard $W_2$ distance satisfies the triangle inequality, and we are simply computing the standard $W_2$ distance on scaled measures.

\begin{definition}
Let $P$ be a transition matrix on on finite state space $\Omega$. Let $\mu$ be some measure on $\Omega$. We define the time-reversal $P'$ of $P$ as a transition matrix where
\begin{align*}
    P'(i,j) = \frac{P(j,i)\mu(j)}{\mu P(i)}.
\end{align*}
\end{definition}

\begin{lemma}
$\mu P P'=\mu$, and $P(j,i)= 0 \implies P'(i,j)=0$.
\label{timereversal}
\end{lemma}
\begin{proof}
It is clear from the form of $P'$ that $P(j,i) = 0\implies P'(i,j)=0$. Moreover,
\begin{align*}
\mu P P'(x) = \sum_{y\in\Omega} \mu P(y) P'(y,x) = \sum_{y\in\Omega} \mu P(y) \frac{P(x,y)\mu(x)}{\mu P(y)} = \mu(x) \sum_{y\in\Omega} P(x,y) = \mu(x). 
\qquad \qedhere
\end{align*}
\end{proof}

Recall the following useful facts about TV distance.
\begin{lemma}
\label{coupling}
Let $\mu,\nu$ be probability measures on a finite set $\Omega$. Then 
\begin{align*}
\dtv{\mu}{\nu} = \inf_{\pi\in C(\mu,\nu)} \Pr_{(\bX,\bY)\sim \pi}[X \neq Y].
\end{align*}
\end{lemma}

\begin{lemma}
\label{whpTV}
Consider some probability space with probability measure $\Pr$.
Let $E$ be some event that occurs with probability $1-\delta$ under $\Pr$.
For all small enough $\delta$, $\dtv{\Pr[\cdot]}{\Pr[\cdot|E]} = O(\delta)$.
\end{lemma}
\begin{proof}
For any event $A$,
\begin{align*}
\Pr[A|E] = \frac{\Pr[A,E]}{\Pr[E]} \leq \Pr[A](1+O(\delta)) \leq   \Pr[A] + O(\delta),
\end{align*}
and also,
\begin{align*}
\Pr[A|E] \geq \Pr[A,E] \geq \Pr[A] - \delta.
\end{align*}
The result follows from the variational characterization of TV distance.
\end{proof}

\begin{lemma}
\label{largest_clique}
With probability $1-o(1)$, every independent set in $\bG \sim G(n,1/2)$ is of size $O(\log n)$.
\end{lemma}
\begin{proof}
Let $k\geq 4\log n$. Observe that the expectation of the number of $k$-independent sets is 
\begin{align*}
    \binom{n}{k} 2^{-\binom{k}{2}} \leq 2^{k\log n - \binom{k}{2}} \leq 2^{k(\log n -\frac{k-1}{2})} \leq 2^{4\log n(-\log n+\frac{1}{2})} = 2^{-\Omega(\log^2 n)}
\end{align*}
By Markov's inequality, the probability that there exists a $k$-clique is at most $2^{-\Omega(\log^2 n)}$. By a union bound over $4\log n \leq k \leq n$, the conclusion follows.
\end{proof}

For any set $A\subseteq \set{0,1}^n$, call it \emph{increasing} if $x\in A, y\geq x$ coordinate-wise implies $y\in A$.
The \emph{FKG} inequality gives a lower bound on the probability of intersections of increasing events under a class of measures known as ``log-supermodular'' measures that includes the uniform measure.

\begin{lemma}[FKG inequality (\cite{alon2016probabilistic} specialized to uniform on $\set{0,1}^N$)]
\label{FKG}
Let $A_1,\dots,A_m\subseteq \set{0,1}^N$ be increasing, and $\Pr[\cdot]$ denote the uniform probability measure on $\set{0,1}^N$. Then
\begin{align*}
\Pr[A_1\cap\cdots\cap A_m] \geq \prod_{i=1}^m \Pr[A_i].
\end{align*}
\end{lemma}

\begin{lemma}[\cite{penev2022combinatorics}]
\label{fact:factorial_bound}
Let $n\geq1$.
Then
\begin{align*}
    \frac{n^{n}}{e^{n-1}} \leq n! \leq \frac{n^{n+1}}{e^{n-1}}
\end{align*}
\end{lemma}
\section{Properties of the hardcore model on an \erdos-\renyi graph}
\label{sec:hardcore}

Let $\hc$ denote the law of the hardcore model with fugacity $\lambda=1$ on $\bG \sim G(n,1/2)$; that is, for each subset $S \subset [n]$, 
\[
\hc(S) \propto \lambda^{|S|} \Ind[S \text{ independent set in }\bG] = \Ind[S \text{ independent set in }\bG].
\]
Since $\bG \sim G(n,1/2)$ is random, $\hc$ is a random measure.

Let $\bZ$ be the partition function of $\hc$, and let $\bZ_k$ denote the contribution to $\bZ$ from $k$-sized independent sets; as $\lambda=1$, $\bZ_k$ is just the number of $k$-independent sets.
In \Cref{lem:measure_approximation}, we will show that for most graphs $\bG \sim G(n,1/2)$, the size of an independent set drawn from $\hc$ is highly concentrated around the value $k^* := \log n - \log\log n$. 
We intend that $k^*$ is an integer (the closest integer to $\log n - \log\log n$), but we suppress floor/ceiling notation where it does not affect clarity to keep notation clean.

Let $\hca$ be $\hc$ conditioned on drawing an independent set of size in the range $[\ks-a,\ks+a]$. 
Let $\hck$ be $\hc$ conditioned on drawing an independent set of size $k$.

\begin{lemma}
\torestate{
\label{lem:measure_approximation}
Let $a\leq \ks$.
Then with probability $1-o(1)$,
\begin{align*}
    \dtv{\hc}{\hca} = O(2^{-a}) + 2^{-\Omega(\log^2 n)}.
\end{align*}}
\end{lemma}

\begin{proof}

The probability of sampling an independent set of size $k$ under $\mu$ is given by the ratio $\bZ_k/\bZ$. 
We will argue that all $\bZ_k$ concentrate well enough around $\EZ_k :=\E[\bZ_k]$ for all $k$, and then argue that the contribution of $k$ outside the window $k^* \pm a$ is negligible.

Let $\bE$ be the event that for all $k$ in the set $L \defeq [0,2\log n - 5\log\log n]$, $|\bZ_k - M_k| < \frac{1}{100} M_k$, and that for all other $k\notin L$, $|\bZ_k - M_k| < n^2 M_k$.
For the larger $k \not\in L$, by Markov's inequality, $\Pr[|\bZ_k - \EZ_k| \ge n^2\EZ_k] \le \frac{1}{n^2}$. 
In \Cref{sec:conc} we will argue via the second moment method that stronger concentration is likely for $k \in L$:
\begin{lemma}
\torestate{
\label{small_clique_concentration}
Let $\epsilon>0, k\leq 2\log n-5\log\log n$.
Then
\begin{align*}
\Pr\left[\left|\bZ_k - \EZ_k \right| \geq \epsilon \EZ_k \right] \leq O\left(\frac{k^5}{\epsilon^2 n^2} \right).
\end{align*}
}
\end{lemma}
\noindent Applying a union bound over $k$, $\Pr[\bE] \ge 1 - O(\frac{\log^6 n}{n^2})$.
Henceforth, we condition on $\bE$.

As we will argue in \Cref{sec:conc},
the sequence $\EZ_k$ is well-behaved, decreasing geometrically around $\ks$:
\begin{lemma}
\torestate{
\label{clique_expectation}
Let $K>0$ be an integer.
Suppose $k=\ks + K$.
Then
\begin{align*}
    \frac{\EZ_{k+1}}{\EZ_k} \leq \left(1 + O\left(\frac{\log \log n}{\log n}\right) \right) 2^{-K}.
\end{align*}
If $k=\ks -K$, then
\begin{align*}
\frac{\EZ_{k-1}}{\EZ_k} \leq \left(1 + O\left(\frac{\log \log n}{n}\right) \right) 2^{-K-1}.
\end{align*}
}
\end{lemma}
So for any constant factor $A < 2$, for $n$ large enough, the sequence $\EZ_k$ decreases geometrically by $A$ as $k$ ascends/descends from $\ks$.
Conditioned on $\bE$, when $k \le 2\log n - 5\log\log n$, $0.99 \le \bZ_k/M_k \le 1.01$, so up to $k$ not too large the $\bZ_k$ sequence is also geometrically decreasing around $\ks$ by a (conservative) factor of $1.5$.

It remains to argue about the independent sets of size at least $j\geq2\log n -5\log \log n$.
 By iteratively applying \Cref{clique_expectation}, we see that 
\begin{align*}
\frac{\EZ_j}{\EZ_{\ks}} &\leq \left(1 + O\left(\frac{\log \log n}{\log n}\right) \right)^{j-\ks} \prod_{i=1}^{j-\ks} 2^{-i}  
\leq 2^{O(j\log \log n/\log n) - \binom{j}{2}} 
%&\leq 2^{O((2\log n - 5\log\log n)\log \log n/\log n) - (2\log n - 5\log\log n)^2/2)} \\
\leq 2^{-\Omega(j^2)}.
\end{align*}
Since $n^{2(j-\ks)} = 2^{O(j \log\log n)} = 2^{o(j^2)}$, under the event $\bE$ we can transfer the upper bound on $\EZ_j/\EZ_{\ks}$ to $\bZ_j/\bZ_{\ks}$:
\begin{align*}
\frac{\bZ_j}{\bZ_{\ks}} \leq  2^{-\Omega(j^2)}.
\end{align*}

Putting the above together, 
\begin{align*}
\sum_{k=0}^n \bZ_k &= \sum_{k=\ks-a}^{\ks+a} \bZ_k + \sum_{k=0}^{\ks-a-1} \bZ_k + \sum_{k=\ks+a+1}^{2\log n - 5\log\log n} \bZ_k + \sum_{k=2\log n - 5\log\log n + 1}^n \bZ_k \\
&\leq \sum_{k=\ks-a}^{\ks+a} \bZ_k + O(2^{-a}\bZ_\ks) + n 2^{-\Omega(\log^2 n)} \bZ_{\ks} \\
&\leq \left(1+O\left(2^{-a} + 2^{-\Omega(\log^2 n)}\right)\right) \sum_{k=\ks-a}^{\ks+a} \bZ_k.
\end{align*}
Let $\delta = O(2^{-a}) + 2^{-\Omega(\log^2 n)}$.
Then we observe that the probability of sampling an independent set in the range $[\ks-a,\ks+a]$ is at least $1-O(\delta)$.
It follows from standard arguments (which we have recorded in \Cref{whpTV}) that $\dtv{\hc}{\hca} = O(\delta)$.
\end{proof}

\subsection{Concentration of number of $k$-independent sets for small $k$}
\label{sec:conc}

In this section we will prove \Cref{small_clique_concentration,clique_expectation}.

\restatelemma{small_clique_concentration}
\begin{proof}
The proof will proceed by the second moment method.
Let $\Ind_S$ be the indicator of the event that $S$ is an independent set.
Then $\bZ_k = \sum_{S\in \binom{[n]}{k}} \Ind_S$.
Then
\begin{align*}
\Var \bZ_k &= \E\left[ \left(\sum_{S\in \binom{[n]}{k}} \Ind_S\right)^2 \right] - \left(\sum_{S\in \binom{[n]}{k}} 2^{-\binom{k}{2}} \right)^2 \\ 
&= \sum_{S,T\in \binom{[n]}{k}} 2^{-2\binom{k}{2} + \binom{|S\cap T|}{2}}-2^{-2\binom{k}{2}} \\ 
&= \sum_{r=0}^k \binom{n}{2k-r}\binom{2k-r}{k}\binom{k}{r}\left(2^{\binom{r}{2}} - 1 \right)2^{-2\binom{k}{2}} \\
&\le \sum_{r=2}^k \binom{n}{2k-r}\binom{2k-r}{k}\binom{k}{r}2^{\binom{r}{2}} 2^{-2\binom{k}{2}},
\end{align*}
where the third equality follows by counting the number of subsets $S,T$ that intersect on $r$ elements.
Call the $r$th term in the summation $T_r$.
For each $r \ge 2$ and $n$ sufficiently large, we have
\begin{align*}
\frac{T_r}{T_2} 
&\le \frac{(2k-2)(2k-3)\cdots(2k-r)}{(n-2k+r)(n-2k+r-2) \cdots (n-2k+3)}\frac{\binom{k}{r}}{\binom{k}{2}} \frac{2^{r(r-1)/2}}{2}
\le \frac{2}{r!} \left(\frac{2k^2 \cdot 2^{(r+1)/2}}{n-2k}\right)^{r-2} < 1,
\end{align*}
where the final inequality follows because $2 \le r \le k \le 2\log n - 5 \log\log n$.
Applying this to the summation above, $\Var \bZ_k \le k \cdot T_2$.

Now, 
\begin{align*}
\frac{T_2}{(\E Z_k)^2}
&= \frac{\binom{n}{2k-2}\binom{2k-2}{k}\binom{k}{2}2^{\binom{2}{2}} 2^{-2\binom{k}{2}}}{\binom{n}{k}^2 2^{-2\binom{k}{2}}} \\
&= \frac{((n-k)!)^2 (k!)^2}{(n-2k+2)! ((k-2)!)^2 n!} \\
&\leq k^4 \frac{(n-k)!^2}{n!(n-2k+2)!} 
= k^4 \frac{(n-k)(n-k-1)\cdots(n-2k+3)}{n(n-1)\cdots(n-k+1)} 
= O\left(\frac{k^4}{n^2}\right)
\end{align*}

where the last line follows because there are $k$ terms in the denominator and $k-2$ terms in the numerator, where each term in the numerator is smaller than each term in the denominator.
So $\Var Z_k / (\E Z_k)^2 = O(k^5/n^2)$, and the lemma follows by Chebyshev's inequality.
\end{proof}

\restatelemma{clique_expectation}
\begin{proof}
First observe that for any $j$, 
\begin{align*}
    \frac{\EZ_{j+1}}{\EZ_j} 
= \frac{\binom{n}{j+1}2^{-\binom{j+1}{2}}}{\binom{n}{j}2^{-\binom{j}{2}}}
=  \frac{n-j}{(j+1)2^j}.
\end{align*}
In the first case, 
\begin{align*}
    \frac{\EZ_{k+1}}{\EZ_k} = \frac{(n-k-1)\log n}{(\log n - \log\log n+K)\cdot n \cdot 2^K} \leq 2^{-K} \frac{\log n}{\log n - \log \log n} \leq 2^{-K} \left(1 + O\left(\frac{\log \log n}{\log n} \right)\right).
\end{align*}
In the second case, 
\begin{align*}
\frac{\EZ_{k-1}}{\EZ_k} = \frac{2^{k-1}k}{n-k+1} = \frac{n}{n-\log n + \log \log n +K} \frac{\log n - \log \log n - K}{\log n} 2^{-K-1} \leq 2^{-K-1} \left(1 + O\left(\frac{\log n}{n} \right)\right).\quad \qedhere
\end{align*}
\end{proof}

\section{The hardcore model on \erdos-\renyi graphs exhibits disorder chaos}
\label{sec:disorder}

In this section, we will show that the hardcore model on $G(n,1/2)$ exhibits \emph{transport disorder chaos}, ultimately proving \Cref{thm:chaos-intro}.

Transport disorder chaos is a property of a distribution over a family of measures; typically, each measure $\mu_G$ in the family is parameterized by some specific object, such as an $n$-vertex graph $G$, and the distribution $\calD$ over these objects induces the distribution over measures.
Disorder chaos further requires a natural notion of a random noise operator that one can apply to the underlying object $G$, where the noise operator fixes $\calD$.
Applying this noise operator to $G$ naturally induces a new random $G'$, and a new measure $\mu_{G'}$.

In our case, the parameterizing objects are graphs sampled from $G(n,1/2)$, and the measure corresponding to a graph $\bG$ is just $\hc$, the hardcore model.
The noise operator is the natural noise operator over the boolean cube $\{0,1\}^{\binom{[n]}{2}}$, which resamples each edge with some specified probability.
Formally, let $\noise_{1-s}(G)$ be the distribution on graphs $\bG'$ given by resampling each edge in the graph $G$ independently with probability $s$. 
Equivalently, every edge of $G$ is flipped with probability $s/2$. 
We say that the family $\mu_{\bG}$ has disorder chaos if 
\begin{align*}
 \inf_{s \in (0,1)} \liminf_{n\to\infty}  \E_{\substack{\bG\sim G(n,1/2)\\\bG' \sim \noise_{1-s}(\bG)}} \left[\w_2(\mu_{\bG},\mu_{\bG'})\right] > 0.
\end{align*}
Intuitively, the presence of disorder chaos means that if we apply any arbitrarily small random perturbation to $G$, the measure $\mu_G$ will typically experience a noticeable perturbation in $\w_2$ transport distance.

In the $\w_2(\mu,\nu)$ distance (\Cref{def:w2}), the random variable $\bX \sim \mu$ is normalized by $m_2(\mu) \defeq \sqrt{\E_{\bX \sim \mu}[\|\bX\|^2]}$.
We require the following technical lemma, which proves that $m_2(\mu_{\bG})$ concentrates well over the randomness of $\bG$.
\begin{lemma}
\label{hp_concentration_of_size}
For any fixed constants $\alpha >0, \beta \in (0,1)$,
\begin{align*}
\Pr_{\bG\sim G(n,1/2)}\left[m_2(\mu_{\bG})^2 \in (1\pm \beta \pm o_n(1))\log n\right] \geq 1-O\left( \frac{\log^6 n}{n^{2+2\alpha}} \right).
\end{align*}
\end{lemma}

We'll prove \Cref{hp_concentration_of_size} after proving the following Theorem, which is a restatement of \Cref{thm:chaos-intro}:

\begin{theorem}
\begin{align*}
 \inf_{s \in (0,1)} \liminf_{n\to\infty}  \E_{\substack{\bG\sim G(n,1/2)\\\bG' \sim \noise_{1-s}(\bG)}} \left[\w_2(\mu_{\bG},\mu_{\bG'})^2\right] = 2.
\end{align*}
\end{theorem}
\begin{proof}
By definition of the normalized transport distance and our (natural) choice to represent $\mu_{G}$ as a measure over $\{0,1\}^n$,
\[
\w_2(\mu_{G},\mu_{G'})^2 
= \inf_{\pi}\E_{(S,S') \sim \pi} \left\|\frac{\Ind_S}{\sqrt{\E_{S \sim \mu_G}\|\Ind_S\|^2}}- \frac{\Ind_{S'}}{\sqrt{\E_{S' \sim \mu_{G'}}\|\Ind_{S'}\|^2}}\right\|^2
= 2 - 2\sup_{\pi}\E_{(S,S') \sim \pi} \frac{|S \cap S'|}{\sqrt{\E_{S\sim \mu_{G}}|S| \cdot \E_{S'\sim\mu_{G'}}|S'|}}.
\]
The upper bound $\E[\w_2(\mu_{\bG},\mu_{\bG'})^2] \le 2$ follows immediately from the non-negativity of set sizes.
For the lower bound, the high-level argument is that for a typical $\bG$, $\hc$ places most of its mass on independent sets of size roughly $\log n$. 
Resampling each edge in an independent set of size $\log n$ with constant probability $s>0$ typically only leaves behind an independent set of size $o(\log n)$.
Therefore, any coupling between $\hc$ and $\hcp$ will have to map most independent sets of size $\log n$ in $\hc$ to nearly disjoint independent sets also of size $\log n$ in $\hcp$.
We now formalize this argument.

We say that a $k$-independent set in $\bG$ "survives" resampling if it contains a sub-independent set of size at least $\log^2 k$ in $\bT_{1-s} \bG$. 
We define a random variable that corresponds to the fraction of $k$-independent sets in $\bG$ that survive the resampling.
Let $\bF_k$ be the fraction of $k$-independent sets that survive resampling in $\bG$. 
Of course, this only makes sense if $\bG$ has a $k$-independent set to begin with. 
Thus, let $\bF_k = 1$ if no $k$-independent sets exist, and otherwise let
$\bF_k = \E_{\bG'} \E_{\bS\sim \mu_{\bG|k}} \Ind_{E_{\bS}}$ if $\bG$ where $E_S$ is the event that $S$ survives (recall that $\mu_{\bG|k}$ is $\mu_{\bG}$ conditioned on sampling a $k$-set). 
Let $A_k$ be the event that there is a $k$ independent set in $\bG$. 

We will now argue that
\begin{align*}
\E_{\bG} \bF_k = \E_{\bG}\left[\E_{\bG'} \E_{\bS\sim \mu_{\bG|k}} \Ind_{E_{\bS}} | A_k\right]\Pr[A_k] + \E_{bG'}[\bF_k |A_k^C]\Pr[A_k^C] = o(1).
\end{align*}
for any $\log n-2\log\log n \leq k \leq \log n +\log \log n$. From \Cref{small_clique_concentration}, we see that the term corresponding to $A_k^C$ is $o(1)$. 
For the $A_k$ term, observe that
\begin{align*}
\E_{\bG}\left[\E_{\bG'} \E_{\bS\sim \mu_{\bG|k}} \Ind_{E_{\bS}} | A_k\right] &= \E_{\bG}\left[ \E_{\bS\sim \mu_{\bG|k}} \E_{\bG'} \Ind_{E_{\bS}} | A_k\right] \\
&= \E_{\bG}\left[ \E_{\bS\sim \mu_{\bG|k}} \Pr[\bS\text{ survives in }\bG'] | A_k\right] \\
&= \Pr_{\bS' \sim \noise_{1-s}(\overline{K_k})}[\bS' \text{ contains an i.s. of size at least } \log^2 k],
\end{align*}
where $\overline{K_k}$ is the empty graph on $k$ vertices.
As $s > 0$ is a constant independent of $n$, by a union bound
\begin{align*}
\Pr_{\bS' \sim \noise_{1-s}\overline{K_k}}[\bS' \text{ contains an i.s. of size at least } \log^2 k] 
&\leq \binom{k}{\log^2 k} \cdot \left(1-\tfrac{s}{2}\right)^{\binom{\log^2 k}{2}} \\ 
&\leq \exp\left(\log^2 k \cdot \ln k - \binom{\log^2 k}{2} \cdot \ln \frac{1}{1-\tfrac{s}{2}}\right)\\
&= \exp(-\Omega(\log^4 k)) 
= o(1).
\end{align*}
Thus $\E_{\bG} \bF_k = o(1)$.

We now argue that $\mu_{\bG},\mu_{\bG'}$ concentrate sufficiently well on $\log n$-sized sets, and that those sets typically do not survive.
Let $\gamma,\delta,\beta>0$ be small constants. 
From \Cref{lem:measure_approximation}, we know that for any $a = O_n(1)$ with probability $1-o(1)$, 
\begin{align*}
    \dtv{\hc}{\hca} = O(2^{-a}) + 2^{-\Omega(\log^2 n)}.
\end{align*}
Thus, by setting $a$ to be a large enough constant, we find that $\Pr_{\bS\sim\hc}[|\bS| \in [\ks-a,\ks+a]] \geq 1-\gamma$. 
Let $E$ be the event that this holds in both $\mu_{\bG}$ and $\mu_{\bG'}$.
Now, let the event $E'$ be the event that $\bF_k < \delta$ for all $k\in[\ks-a,\ks+a]$. 
By applying Markov's inequality for each $k$ and then a union bound, $\Pr[E'] = 1-o(1)$. 
Also, let $E''$ be the event that $\E_{S \sim{\hc}} \norm{S}_2^2, \E_{S' \sim \mu_{\bG'}} \norm{S'}_2^2 \in (1 \pm \beta) \log n$. 
By \Cref{hp_concentration_of_size}, $\Pr[E''] = 1-o(1)$. 
Finally, let $E'''$ be the event that $\bG,\bG'$ have cliques only of size $O(\log n)$; by \Cref{largest_clique}, this happens with probability $1-o(1)$.
Thus, $\Pr[E\cap E' \cap E'' \cap E'''] = 1-o(1)$.

Condition on $E\cap E'\cap E'' \cap E'''$. 
We will now analyze any possible coupling $\pi=\pi(\bG,\bG')$ of $\mu_{\bG},\mu_{\bG'}$. 
Let $L = [\ks-a,\ks+a]$. 
Observe that $1-\gamma$ of the proportion of independent sets of $\bG$ and $\bG'$ have size in $L$. 
Moreover, a $(1-\gamma)(1-\delta)$ proportion of the independent sets in $\bG$ have both size in $L$ and no independent set of more than size $(1+\beta)(\log\log n)^2$ contained inside it. 
Let this set of independent sets in $\bG$ be $\calS$. 
As just argued, $\calS$ has probability mass at least $(1-\gamma)(1-\delta)$; $\pi$ can assign at most $\gamma$ of the mass of $S$ to independent sets in $\bG'$ with size outside of $L$, so $(1-\gamma)(1-\delta)-\gamma$ of the mass in $\calS$ is on independent sets of size in $L$ in $\hcp$. 
Thus, $1-\eta \defeq (1-\gamma)(1-\delta)-\gamma$ of the mass of independent sets in $\calS$ is coupled to independent sets in $\bG'$ with at most $(1+\beta)(\log\log n)^2$ shared nodes. 
Therefore,
\begin{align*}
\w_2(\mu_{\bG}, \mu_{\bG'})^2 
\geq 2 - 2\frac{(1-\eta)\cdot(1+\beta)(\log\log n)^2 + \eta\cdot O(\log n)}{(1\pm \beta)\log n} \geq 2 - o(1) - \frac{\eta}{1+\beta}
\end{align*}
The above inequality holds for the expectation as well because $\Pr[E\cap E' \cap E''] = 1-o(1)$. Thus,
\begin{align*}
 \liminf_{n\to\infty}  \E_{\substack{\bG\sim G(n,1/2)\\\bG' \sim \noise_{1-s}(\bG)}} [\w_2^2(\mu_{\bG},\mu_{\bG'})] \geq 2 - \frac{\eta}{1+\beta}.
\end{align*}
Since this is true for every small enough $\gamma,\delta,\beta>0$, we can take $\eta$ arbitrarily small. 
As the above applies to any constant $s > 0$, the conclusion follows.
\end{proof}

We now prove \Cref{hp_concentration_of_size}.

\begin{proof}[Proof of \Cref{hp_concentration_of_size}]
We must establish the concentration over $\bG$ of $m_2(\mu_{\bG})^2 = \E_{\bS\sim\hc}\|\Ind_{\bS}\|^2 = \E_{\bS \sim \hc}|\bS|$.
Recall again that $\Pr[|\bS| = k] = \bZ_k/\bZ$, and we have defined $M_k = \E_{\bG}[\bZ_k]$.
By setting $\eps=n^\alpha$ in \Cref{small_clique_concentration}, we see that for all $0\leq k\leq 2\log n - 5\log \log n$,
\begin{align*}
\Pr[|\bZ_k-\EZ_k| \geq n^\alpha \EZ_k] \leq O\left(\frac{\log^5 n}{n^{2+2\alpha}}\right).
\end{align*}
Furthermore, observe that $\Pr[|\bZ_k-\EZ_k| \geq n^4 \EZ_k]\leq 1/n^4$ by Markov's inequality for all $k>2\log n-5\log \log n$. 
Let the event that none of these events occur be $E$; observe that by a union bound, $\Pr[E] \geq 1-O((\log^6 n) / n^{2+2\alpha})$. Henceforth, we condition on $E$.

By \Cref{clique_expectation},
\begin{align*}
\frac{\EZ_j}{\EZ_{\ks}} &\leq \left(1 + O\left(\frac{\log \log n}{\log n}\right) \right)^{j-\ks} \cdot \prod_{i=1}^{j-\ks} 2^{-i}  
\leq 2^{O(j\log \log n/\log n) - \binom{j-\ks}{2}}
\leq 2^{-(j-\ks)^2/2 + O(j\log\log n /\log n)}.
\end{align*}
so for any $\ks+\beta\log n \leq j \leq 2\log n-5\log \log n$, we see that $\EZ_{j}/\EZ_{\ks} \leq 2^{-\beta^2 \log^2 n/2 + O(\log\log n)}$. 
The same bound holds for any $j\leq \ks-\beta\log n$ using similar arguments. 
By using the same computation above, any $j>2\log n - 5\log \log n$, $\EZ_j/\EZ_{\ks}\leq 2^{-\Omega(\log^2 n)}$. 
Under $E$, the same bounds hold in each case for $\bZ_j/\bZ_{\ks}$. Let $A = (\ks-\beta\log n, \ks + \beta\log n)$.
\begin{align*}
\Pr_{\bS \sim \mu_{\bG}}[|\bS|\notin A] = \sum_{j\notin A} \Pr[|\bS| = j] \leq \sum_{j\notin A} 2^{-\Omega(\beta^2\log^2n)}\Pr[|\bS|=\ks] \leq 2^{-\Omega(\beta^2\log ^2 n)}.
\end{align*}
Our conclusion follows by an averaging argument.
On the one hand,
\begin{align*}
m_2({\hc})^2 
= \E_{\bS \sim \mu_{\bG}} |\bS|
&\le (\ks + \beta \log n) + (\max_{S} |S|)\cdot \Pr[|\bS| > \ks + \beta n]\\
&\le \ks + \beta \log n + n \cdot 2^{-\beta^2 \log^2 n/2 + O(\log\log n)}
= (1 + \beta + o_n(1))\log n,
\end{align*}
And on the other,
\begin{align*}
m_2({\hc})^2 
= \E_{\bS \sim \mu_{\bG}} |\bS|
&\ge (\ks - \beta \log n)\Pr[|\bS| \ge \ks - \beta\log n] \\
&\ge (\ks - \beta \log n)(1 - 2^{-\beta^2 \log^2 n + O(\log n\log\log n)})
= (1 - \beta - o_n(1))\log n. \qedhere
\end{align*}
\end{proof}
\section{Glauber Dynamics samples in transport distance}
\label{sec:glauber}

In this section we will prove \Cref{thm:glauber-sample}.
We'll introduce the Glauber Dynamics and show that it can be coupled with a simple greedy algorithm.
Then, we'll show that the greedy algorithm is close to uniform over independent sets of size $k^- = \log n - O(\log\log n)$.
Finally, we'll show that there is a good $\w_2$ coupling between the uniform distribution on independent sets of size $k^-$ and $\mu_{\bG}$.

\subsection{Coupling the Glauber Dynamics with the \greedy algorithm}

\Cref{thm:glauber-sample} states that Glauber dynamics (\Cref{alg:glauber}) (run until a stopping condition is reached) samples in $\w_2$ from $\hc$. 
However, we will find it more convenient to analyze a randomized greedy algorithm (\Cref{alg:gradient_ascent}, henceforth known as \greedy). In this section, we state both algorithms and show that they can be coupled with high probability.

For a measure $\mu$, recall that the generic Glauber dynamics is a Markov chain where at any state $\sigma \in \set{0,1}^n$, we pick a coordinate $\bi\sim[n]$ uniformly at random, and then resample from the conditional measure of $\mu$ where we enforce that for our new sample $X$, $X_j=\sigma_j$ for all $j\neq i$. 
It is easy to check that the stationary distribution of this Markov chain is $\mu$ (e.g. see \cite{LP17}, chapter 3).
For the hardcore model on any graph $G$ with $\lambda=1$, it is easy to check that \Cref{alg:glauber} run on graph $G$ exactly corresponds to the generic Glauber dynamics with measure $\mu_G$. 
Therefore, the stationary measure of \Cref{alg:glauber} run on $G$ is precisely $\mu_G$ (although we do not ever use this fact in our analysis).

\greedy builds an independent set, starting from $\emptyset$ and iteratively adding a uniformly random vertex as long as one is available to add.
The only difference between this algorithm and Glauber dynamics is that Glauber dynamics samples a vertex uniformly from $[n]$ at each step and decides whether to add or remove it; removal happens so rarely that it is easy to construct a coupling of \greedy and Glauber dynamics.
This coupling will allow us to translate results from \greedy to Glauber dynamics.

\begin{algorithm}
\caption{Glauber dynamics}\label{alg:glauber}
\begin{algorithmic}[1]
\Require A graph $G$, a size $s$, and a time $T$
\Ensure FAIL or an independent set of size $s$
\State Let $S_0 = \emptyset$.
\For{$t = 1$ \textbf{to} $T$}
    \State $v \sim \text{Unif}([n])$
    \If{$S_{t-1}\cup v$ is not an independent set}
        \State $S_t = S_{t-1}$
    \EndIf
    \State $S_t \sim \text{Unif}\set{S_{t-1}\cup v, S_{t-1}\backslash v}$
    \If{$|S_t|=s$}
        \State Return $S_t$
    \EndIf
\EndFor
\State Return FAIL
\end{algorithmic}
\end{algorithm}

\begin{algorithm}
\caption{\greedy}\label{alg:gradient_ascent}
\begin{algorithmic}[1]
\Require A graph $G$ and size $s$
\Ensure FAIL or an independent set of size $s$
\State Let $S_0 = \emptyset$
\For{$i = 1$ \textbf{to} $s$}
    \State $X_i = \set{v \in [n]\backslash S \mid v \cup S \text{ is an independent set in }G}$
    \If{$X_i = \emptyset$}
        \State \Return FAIL
    \EndIf
    \State $v_i \sim \text{Unif}(X_i)$
    \State $S_i = S_{i-1} \cup v_i$
\EndFor
\State \Return $S_s$
\end{algorithmic}
\label{gradient_ascent}
\end{algorithm}

For a fixed graph $G$, size $s>0$ and time $T>0$, let $\Glauber(G,s,T), \greedy(G,s)$, denote the random variable corresponding to the output of \Cref{alg:glauber}, \Cref{alg:gradient_ascent} respectively.

In both Glauber Dynamics and \greedy, it will be important to us to understand how many vertices are not neighbors with any vertex of $S_t$ (and thus can be added to $S_t$).
This motivates the following definition:
\begin{definition}
Given a graph $G=(V,E)$ and any $S\subseteq V$, the up-degree of $S$ is defined as 
\begin{align*}
    \mathrm{deg}^{\uparrow}(S) = |\set{v\in V \backslash S \mid (u,v)\notin E(G) \text{ for all } u\in S}|.
\end{align*}
\end{definition}
\noindent Note that the up-degree is also a function of the graph (though the notation does not explicitly reflect this).

Let $d_j = \E[\updeg([j])]$.
By linearity of expectation, $d_j = (n-j)2^{-j}$.
A combination of Chernoff and union bounds provide simultaneous concentration of $\updeg(S)$ around $d_{|S|}$ for all $S$ not too large.

\begin{lemma}
\label{degree_concentration}
Let $\bG\sim G(n,1/2), C>6$.
Then with probability $1-o(1)$, all sets $S\subseteq[n]$ of size $0,\dots,\log n - C\log \log n$ satisfy $|\updeg(S) - d_{|S|}| \leq d_{|S|}^{2/3}$.
\end{lemma}
\begin{proof}
Fix a set $S$ of size $j$.
Then observe that $\updeg(S)$ is distributed as $\Binom(n-j,2^{-j})$.
So by a standard Chernoff bound,
\begin{align*}
\Pr[|\updeg(S) - d_j| \geq \delta_j d_j] \leq 2\exp(-\delta_j^2 d_j/3).
\end{align*}
Set $\delta_j=d_j^{-1/3}$.
There are $\binom{n}{j} \leq 2^{j\log n}$ subsets of size $j$.
The union bound tells us that 
\begin{align*}
\Pr\left[\exists S \in [n]^{\le \log n - C \log\log n}\, \text{ s.t. } \frac{|\updeg(S) - d_{|S|}|}{d_{|S|}} \geq \delta_{|S|}\right] &\leq \sum_{j=0}^{\log n - C\log \log n} \binom{n}{j} \exp(-\delta_j^2d_j/3) \\
&\leq \sum_{j=0}^{\log n - C\log \log n} \exp(\Theta(j\log n) -d_j^{1/3}/3).
\end{align*}
Let the $j$th term be $T_j$.
Observe that $T_j$ is increasing because $j\log n$ is increasing in $j$, and $d_j$ is decreasing in $j$.
Thus it suffices to understand the last term.
Let $k=\log n - C\log \log n$.
Then
\begin{align*}
d_k = (n-k)2^{-k} = \frac{n-k}{n}\log^Cn.
\end{align*}
Moreover,
\begin{align*}
\Theta(k\log n) = \Theta(\log^2 n).
\end{align*}
Therefore, if $C>6$ is a constant, then $T_k = \exp(-\Omega(\log^{C/3-2} n)) = o(1)$, so we are done.
\end{proof}

Also, we will introduce some notation that we will need to analyze Glauber dynamics. 
Recall from the description of \Cref{alg:glauber} that the state at iteration $t$ is the set $\bS_t$. Let the stopping time 
\begin{align*}
\bT_i = \inf_{t\geq 0}\set{|\bS_t|=i}
\end{align*}
denote the first time that Glauber dynamics reaches a set of size $i$.

\begin{theorem}
\label{coupling_algorithms}
Fix $0<s<\log n - C\log \log n, C>6, T\geq100 \cdot 2^s \cdot\log n$. With probability $1-o(1)$ over the graph $\bG \sim G(n,1/2)$, there exists a coupling $\pi =\pi(\bG)$ of $\Glauber(\bG,s, T), \greedy(\bG,s)$ such that
\begin{align*}
\Pr_{\pi}[\Glauber(\bG,s) = \greedy(\bG,s)] = 1-o(1).
\end{align*}
\end{theorem}
\begin{proof}
We will first provide intuition for why such a coupling is possible.
Glauber dynamics chooses a vertex $i$ uniformly from $[n]$ at each step; if the vertex is in the independent set it is (stochastically) dropped, and if it is not in the independent set then it is (stochastically) added so long as it extends the independent set. 
The reason we are able to couple Glauber dynamics and \greedy is because so long as the independent set has size $\le s =\log n - 2 \log\log n$, there are many more vertices that could be added to the independent set than vertices within the independent set itself, so it is unlikely that Glauber dynamics drops vertices before $s$ are added. 
Moreover, because it does not take too long to add a vertex at any step, with high probability, Glauber finds a set of size $s$ in time $T$ as above.

With the above in mind, we will construct our coupling.
For $i=0,\dots,s$, let $E_{i}$ be the event that for all $t\in [\bT_{i-1},\bT_{i})$, $|\bS_t| = i-1$.
Thus, if $\cap_{i=1}^s E_i$ occurs, then up to reaching size $s$, Glauber dynamics does not ever drop any vertices.
Let $D_i$ be the event that beginning from $\bT_{i-1}$, the time it takes for Glauber to add or drop a vertex is at most $(10\log n)n/d_i$; therefore, if $\cap_{i=1}^s E_i$ and $\cap_{i=1}^s D_i$ occurs, it takes at most
\begin{align*}
    \bT_s \leq \sum_{i=0}^{s-1} \bT_{i+1}-\bT_i \leq 10\log n \sum_{i=0}^{s-1} \frac{n}{n-i}2^i \leq  30\cdot 2^s \cdot \log n
\end{align*}
steps of Glauber dynamics to reach an independent set of size $s$.
Our coupling will then consist of running \greedy as usual, where we let $\bv_1,\dots,\bv_s$ be the sequence of vertices it adds.
If $(\cap_{i=1}^s E_i) \cap (\cap_{i=1}^s D_i)$ occurs, which we claim happens with probability $1-O(\log^{2-C}n)$, we let the vertex added at time $\bT_i$ be $\bv_i$, so that the outputs of \greedy and Glauber dynamics are the same. 
If it does not occur, we just run Glauber dynamics and \greedy as usual independently of each other.
Observe that each $\bv_i$ is chosen uniformly at random from the set of nodes that form an independent set with $\set{\bv_1,\dots,\bv_{i-1}}$, so this coupling results in Glauber dynamics being distributed as it should be.

We now justify that $\Pr[(\cap_{i=1}^s E_i) \cap (\cap_{i=1}^s D_i)] = 1-O(\log^{2-C}n)$ with probability $1-o(1)$ over the randomness of $\bG$. 
First, we condition on \Cref{degree_concentration}, obtaining that all sets $S\in[n]^{\leq \log n - C\log\log n}$ satisfy $|\updeg(S) - d_{|S|}| \leq d_{|S|}^{2/3}$. 
This occurs with probability $1-o(1)$ over the graph $\bG$.
Because there exist independent sets of size $s$, $\bT_0,\dots,\bT_s<\infty$ almost surely.
We are interested in showing that $\Pr[E_{i+1}^C | \bS_{T_{i}}]$ is close to $0$. 
This is the event that starting from a set of size $i$, we drop a vertex before adding one. 
The number of vertices that could be added is $\updeg(\bS_{\bT_{i}}) \geq d_{i}(1 - d_i^{-1/3})$, where $d_i = (n-i)2^{-i}$ and the number of vertices that we can drop is $i$. 
Therefore, the probability that we drop a vertex before adding one is at most
\begin{align*}
i/(i+d_{i}(1 - d_i^{-1/3})) \leq 2i/d_i \leq 4 \log^{1-C}n
\end{align*}
where the last inequality follows because $i/d_i$ is increasing in $i$. 

To see that $\Pr[D_{i+1}^C | \bS_{\bT_i}]$ is small, observe that we are just interested in the probability that $(10\log n)n/d_i$ i.i.d. coins with heads probability at least $d_i(1-d_i^{-1/3})/n$ are all tails. This is just
\begin{align*}
\Pr[D_{i+1}^C | \bS_{\bT_i}] \leq \left(1 - \frac{d_i(1-d_i^{-1/3})}{n}\right)^{(10\log n) n/d_i} \leq n^{-9}.
\end{align*}

By a union bound, it is clear that $\Pr[(\cup_{i=1}^s E_i^C)\cup(\cup_{i=1}^s D_i^C)] \leq 5\log^{2-C}n$. 
\end{proof}

\subsection{\greedy samples largeish cliques in total variation}

We will show that, with high probability over $G$, \greedy can approximately sample an independent set in total variation distance from $\hc$ restricted to independent sets of size $\log n - C\log \log n$ for large enough $C$. 

We explain some high-level reason that \greedy samples up to size $\log n - C\log\log n$. 
Fix any independent set $S$ of size $s$; without loss of generality suppose $S = \{1,\ldots,s\}$.
Using the notation of \Cref{gradient_ascent}, \greedy outputs $S$ if and only if the sequence $v_1,\ldots,v_s$ is some permutation of the elements of $S$.
The probability of seeing any given permutation $\pi$ of the elements of $S$ is proportional to $\prod_{i=1}^{s-1} |X_i|^{-1}$, where $X_i$ is the set of all vertices that is not adjacent to any vertex in $\set{\pi(1),\dots,\pi(i)}$. 
Note that $|X_i| = \updeg(S)$, and from \Cref{degree_concentration} we have strong concentration for $\updeg(S)$ for any set $S$ not too large.
As long as $k$ is not too large, we have that $|X_1|,\ldots,|X_{k-1}|$ all concentrate simultaneously, which is enough to prove our theorem.

\begin{theorem}
\label{sampling_small}
Let $\bG\sim G(n,1/2),C>6,k=\log n - C\log \log n$, and let $\alg:=\alg(\bG)$ denote the probability measure associated with running the algorithm \Cref{gradient_ascent} on $\bG$ for time $k$.
Then with probability $1-o(1)$ over the randomness of $\bG$, $\dtv{\alg}{\hck} = o(1)$.
\end{theorem}

\begin{proof}
Let $C>6$.
\Cref{degree_concentration} tells us that with probability $1-o(1)$, for any $S\in [n]^{\leq \log n-C\log \log n}$, $|\updeg(S) - d_{|S|}| \leq d_{|S|}^{2/3}$. In the remainder of the proof, we condition on this event.
By the concentration of the up-degree, no matter what path \greedy takes, we will output some independent set of size $k$.
We will analyze the probability of outputting any particular $k$-independent set.

Let $\bS$ be an arbitrary $k$-independent set in $\bG$.
To output $\bS$, the $v_1,\dots, v_k$ in the description of the algorithm in \Cref{alg:gradient_ascent} must be a permutation of the vertices in $\bS$.
Fix a permutation and call the vertices $v_1,\dots,v_k$.
We will now understand the probability of the algorithm inducing such a permutation.
Since we have conditioned on concentration of the up-degrees,
\begin{align*}
p\defeq \Pr[\text{\greedy produces }v_1,\dots,v_k] \in \prod_{j=0}^{k-1} \frac{1}{d_j(1\pm d_j^{-1/3})}.
\end{align*}
We will now show that $p$ is close to $\prod_{j=0}^{k-1}d_j^{-1}$, so that the probability of reaching any $k$-set through any given permutation is similar. 
Let $D = \prod_{j=0}^{k-1} d_j$. Observe that the possible multiplicative error of $p$ from $1/D$ is bounded as
\begin{align*}
\prod_{j=0}^{k-1} \frac{1}{1 + d_j^{-1/3}} \leq \prod_{j=0}^{k-1} \frac{1}{1\pm d_j^{-1/3}} \leq \prod_{j=0}^{k-1} \frac{1}{1 - d_j^{-1/3}}.
\end{align*}
There is a negative neighborhood of 0 such that $1/(1+x)\leq e^{-2x}$. 
Then
\begin{align*}
\prod_{j=0}^{k-1} \frac{1}{1 - d_j^{-1/3}} &\leq \exp\left(2 \sum_{j=0}^{k-1}  d_j^{-1/3}\right) = \exp\left(2 \sum_{j=0}^{k-1}  \frac{2^{j/3}}{(n-j)^{1/3}}\right) \\
&\leq \exp\left(\frac{3}{n^{1/3}} \sum_{j=0}^{k-1}  2^{j/3}\right) \leq \exp\left(\frac{3}{n^{1/3}} O(2^{k/3})\right) = \exp\left(O\left(\frac{1}{\log^{C/3}n}\right)\right) \\
&= 1+O\left(\frac{1}{\log^{C/3}n}\right).
\end{align*}
By a very similar argument, we see that 
\begin{align*}
    \prod_{j=0}^{k-1} \frac{1}{1 + d_j^{-1/3}} \geq 1-O\left(\frac{1}{\log^{C/3}n}\right).
\end{align*}
So $p \in (1\pm o(1))D^{-1}$.
By summing over all possible permutations of $S$, $\bp_S\defeq \Pr[S] \in (1\pm o(1))k!/D$.
Because the algorithm terminates successfully, 
\begin{align*}
1  = \sum_{S\text{ i.s. in }\bG} \bp_S = \bZ_k \cdot (1\pm o(1))\frac{k!}{D},
\end{align*}
Thus $1/\bZ_k = (1\pm o(1)) k!/D = (1\pm o(1)) \bp_S$ for all independent set $S$, implying 
$\dtv{\alg}{\hck} = o(1)$ as desired.
\end{proof}

\subsection{Sampling in transport distance}
We have seen that Glauber Dynamics, run for an appropriate amount of time, will sample an approximately uniform independent set of size up to $\km = \log n - C\log \log n$ for $C>6$.
This is not quite sufficient for sampling from $\hc$ in TV distance, as $\hc$ is concentrated on independent sets of size $\ks \defeq \log n - 1 \cdot \log\log n$ (\Cref{lem:measure_approximation}).
Here we will show that sampling from $\mu_{\bG\mid \km}$ in TV distance is produces an approximate sample from $\hc$ in $\w_2$.

To show this, we will produce an $\ell_2$-coupling between the uniform distribution over independent sets of size $\km$ and $\hc$.
At a high level, we will show that each $\km$ independent set is contained in approximately the same number of $k$-independent sets, for every $k\in \ks\pm\log\log n$.

We will first generalize our definition of up-degree to allow extension of an independent set by multiple vertices:

\begin{definition}
Given a graph $G=(V,E)$ and $S\subseteq V$, the $\ell$-up-degree of $S$ is defined as 
\begin{align*}
    \mathrm{deg}^{\uparrow}_\ell(S) = \left|\left\{T \in  \binom{V\backslash S}{\ell}; (u,v)\notin E \text{ for all } u\in S, v\in T, T\text{ is an i.s.}\right\}\right|.
\end{align*}
For any $\ell$-set $T$ that satisfies the above condition, we will say that $T$ completes $S$.
\end{definition}
We will show that for any fixed constant $B$, when $\ell = B\log \log n$, the $\ell$-up-degree of any set $S$ of size $\km$ concentrates well when $\bG \sim G(n,1/2)$ for $n$ sufficiently large.
We will do this by bounding a poly-logarithmic moment of $\updeg_\ell(S) - \E \updeg_\ell(S)$. 

\begin{lemma}
\label{updeg_concentration}
Let $2\leq B\leq C, C>20$, $\ell=B\log\log n$, $\km=\log n-C\log\log n$. Then, for any $0\leq \epsilon \leq 1$ and any set $S\in \binom{[n]}{\km}$, 
\begin{align*}
    \Pr[|\updeg_\ell(S) - \E \updeg_\ell(S)| \geq \epsilon \E \updeg_\ell(S)] \leq \left(\frac{1}{\eps \log n}\right)^{\Omega(\log^5 n /\log\log n)}
\end{align*}
where $\Omega(\cdot)$ hides a positive constant depending on $B,C$.
\end{lemma}

\begin{proof}
Let $m = n-\km$, the number of vertices remaining in $\bG$ after a $\km$-set is removed.
Fix some $\km$-set $S$; without loss of generality, let $S=\set{m+1,\dots,n}$.
Let $\bX = \updeg_\ell(S)$, and let $\mu = \E \bX$.
For any even $d$, by Markov's inequality,
\begin{align*}
\Pr[|\bX-\mu| > \epsilon \mu] \leq \frac{\E [(\bX-\mu)^d]}{\epsilon^d \mu^d}
\end{align*}
We first compute a lower bound for $\mu$:
\begin{align*}
\mu 
= \binom{m}{\ell} 2^{-\ell \km - \binom{\ell}{2}} 
= \frac{m^\ell}{\ell!}2^{-\ell \km - \binom{\ell}{2}} \prod_{j=0}^{\ell-1} \left(1-\frac{j}{m}\right) 
\geq \frac{m^\ell}{\ell!}2^{-\ell \km - \binom{\ell}{2}} e^{-\sum_{j=0}^{\ell-1} \frac{j/m}{1-j/m}} 
\geq \frac{m^\ell}{\ell!}2^{-\ell \km - \binom{\ell}{2}} e^{-\frac{2}{m} \binom{\ell}{2}} 
\end{align*}
because $1-x\geq e^{-x/(1-x)}$ for $x\in[0,1]$ and $\ell=o(m)$.

For any $\ell$-set $A$, let $\Ind_A := \Ind[A \text{ completes }S]$.
Recall that $\bX = \sum_{A\in \binom{[m]}{\ell}} \Ind_A$.
Let $p = \E \Ind_A$.
Then
\begin{align*}
\E [(X-\mu)^d] 
= \E\left[\sum_{A_1,\dots,A_d \in \binom{[m]}{\ell}} \prod_{i=1}^d (\Ind_{A_i} - p) \right] 
= \sum_{A_1,\dots,A_d \in \binom{[m]}{\ell}}\E\left[ \prod_{i=1}^d (\Ind_{A_i} - p) \right],
\end{align*}
Observe that $\E[\Ind_{A}-p] = 0$.
Therefore, for any choice of $A_1,\dots,A_d$, if any one of the $A_i$s is disjoint from the rest, by independence, $\E[\prod (\Ind_{A_i}-p)]=0$.
By a simple counting argument (\Cref{vertex_overlap}), if $|A_1\cup\cdots\cup A_d| > d\ell - d/2$, then such a disjoint set exists.
Thus we only need to consider $A_1,\dots, A_d$ such that $|A_1\cup\cdots\cup A_d|\leq d\ell -d/2$. 

Now, fixing any $A_1,\dots, A_d$, we claim that
\begin{align*}
\E\left[ \prod_{i=1}^d (\Ind_{A_i} - p) \right] 
= 2^d \Pr[A_1,\dots, A_d\text{ complete } S].
\end{align*}
This is because there are $2^d$ terms in the expansion of the LHS, half of which are negative.
We can bound each positive product of probabilities above by $\Pr[\cap_{i=1}^d A_i]$ by \Cref{FKG} because each $A_i$ is an increasing event. 

It then follows that
\begin{align*}
\E [(\bX-\mu)^d]  \leq 2^d \sum_{a=d/2}^{d\ell-\ell} N_a \left(\frac{1}{2} \right)^{(d\ell-a)\km + d\binom{\ell}{2}-\min\set{\frac{a}{\ell}\cdot\binom{\ell}{2},\binom{a}{2}}},
\end{align*}
where $N_a$ is the number of possible ways to choose $A_1,\dots,A_d\in\binom{[m]}{\ell}$ such that $|A_1\cup\cdots\cup A_d| = d\ell-a$.
This is because each of the $d\ell-a$ vertices in the union need to be non-neighbors of every vertex in $S$, and every "internal" edge within $A_1,\dots, A_d$ needs to not exist.
In \Cref{edge_bound}, we argue that there are at least $d\binom{\ell}{2} - \min\set{\frac{a}{\ell}\cdot\binom{\ell}{2} , \binom{a}{2}}$ of these internal edges.
In \Cref{lem:config_bound}, we prove that $N_a \le \binom{m}{d\ell-a} \binom{d\ell-1}{a-1} \frac{(d\ell)!}{(\ell!)^d}$.

We now combine the upper and lower bounds with Markov's inequality: 
\begin{align}
\frac{\E [(X-\mu)^d] }{\mu^d} 
&\leq \frac{2^d \sum_{a=d/2}^{d\ell-\ell} \binom{m}{d\ell-a} \binom{d\ell-1}{a-1}\frac{(d\ell)!}{(\ell!)^d} \left(\frac{1}{2} \right)^{(d\ell-a)\km + d\binom{\ell}{2}-\min\set{\frac{a}{\ell}\binom{\ell}{2},\binom{a}{2}}}}{\frac{m^{d\ell}}{\ell!^d}\left(\frac{1}{2}\right)^{d\ell \km + d\binom{\ell}{2}} e^{-\frac{2d}{m} \binom{\ell}{2}} } \nonumber \\
&= 2^d e^{\frac{2d}{m} \binom{\ell}{2}} \sum_{a=d/2}^{d\ell-\ell}\frac{1}{m^{d\ell}} \binom{m}{d\ell-a} \binom{d\ell-1}{a-1}(d\ell)! 2^{a\km +\min\set{\frac{a}{\ell}\binom{\ell}{2},\binom{a}{2}}}\nonumber\\
&\leq 2^d e^{\frac{2d}{m} \binom{\ell}{2}}\sum_{a=d/2}^{d\ell-\ell} \frac{1}{m^{d\ell}}\frac{m^{d\ell-a}}{(d\ell-a)!} \frac{(d\ell-1)^{a-1}}{(a-1)!} \frac{(d\ell)^{d\ell+1}}{e^{d\ell-1}} 2^{a\km +\min\set{\frac{a}{\ell}\binom{\ell}{2},\binom{a}{2}}},\label{eq:pause}
\end{align}
where in the final line we have applied Stirling's inequality and the upper bound $x!/(x-y)! \le x^y$ as well as \Cref{fact:factorial_bound}.
We can bound the ratio:
\begin{align*}
\frac{(d\ell)^{d\ell+1}}{(d\ell-a)!} &\leq \frac{(d\ell)^{d\ell+1}e^{d\ell-a}}{(d\ell-a)^{d\ell-a}} = e^{d\ell-a} (d\ell)^{a+1} \left(\frac{d\ell}{d\ell-a} \right)^{d\ell-a} \\
&= e^{d\ell-a} (d\ell)^{a+1} \left(1+\frac{a}{d\ell-a} \right)^{d\ell-a} \leq e^{d\ell-a} (d\ell)^{a+1} e^a = e^{d\ell}(d\ell)^{a+1}.
\end{align*}
Continuing from \Cref{eq:pause} above and applying Stirling's inequality and dropping the $(a-1)!$ term,
\begin{align*}
\frac{\E [(\bX-\mu)^d] }{\mu^d}
&\leq 2^d e^{\frac{2d}{m} \binom{\ell}{2}}\sum_{a=d/2}^{d\ell-\ell} m^{-a} (d\ell-1)^{a-1} e^{-d\ell+1} e^{d\ell}(d\ell)^{a+1} 2^{a\km +\min\set{\frac{a}{\ell}\binom{\ell}{2},\binom{a}{2}}} \\
&\leq 2^d e^{\frac{2d}{m} \binom{\ell}{2}+1}\sum_{a=d/2}^{d\ell-\ell} m^{-a} (d\ell)^{2a} 2^{a\km +\min\set{\frac{a}{\ell}\binom{\ell}{2},\binom{a}{2}}}\\
&= 2^d e^{\frac{2d}{m} \binom{\ell}{2}+1}\sum_{a=d/2}^{d\ell-\ell}  \left(\frac{(d\ell)^2 2^{\km + \min\set{\frac{1}{\ell}\binom{\ell}{2},\frac{a-1}{2}}}}{m}\right)^a
\intertext{
Now using that $\km = \log n - C \log\log n$, $\ell = B \log\log n$, and letting $d=\log^Dn/\ell$,
}
&\leq e 2^d (1+o(1))\sum_{a=d/2}^{d\ell-\ell} \left(\frac{n}{m}\cdot 2^{\frac{\ell-1}{2}}\log^{2D-C}n\right)^a\\
&\leq e 2^d (1+o(1))\sum_{a=d/2}^{d\ell-\ell} \left(\left(1 + O\left(\tfrac{\log n}{n}\right)\right)\cdot \log^{2D-C + B/2}n \right)^a\\
\intertext{So long as $2D + B/2 - C < 0$ and is bounded away from $0$, the summation is geometrically decreasing, dominated by the term $a = d/2$.
Thus,}
&\le (1+o(1)) e 2^d \left(e \log^{2D - C + B/2} n\right)^{d/2}\\
&\le (1+o(1)) e \left(4 e \log^{2D-C +B/2} n \right)^{d/2}\\
&\leq \exp\left(-F\log^D n  + O(\log^D n /\log\log n)\right).
\end{align*}
where $F > 0$ is some constant that depends on $B,C,D$.
The result follows by letting $D=5$, after which the conditions $B \le C$ and $C > 20$ ensure that $2D+B/2-C$ is strictly negative.
\end{proof}

We are now ready to prove the following theorem, which completes the proof of \Cref{thm:glauber-sample} (since the output distribution of $\greedy, \alg$, is close in total variation to $\hc$ by \Cref{coupling}, and since $\w_2$ is bounded by $O(1)$).
\begin{theorem}
Recall the definition of $\alg$ in \Cref{sampling_small}. With probability $1-o(1)$,
\begin{align*}
\w_2(\alg, \hc) = o(1)
\end{align*}
\end{theorem}

\begin{proof}
We condition on the event that every independent set in $\bG$ is of size $O(\log n)$ from \Cref{largest_clique}, the event that $\dtv{\hc}{\hca} = o(1)$ for $a=\log\log n$ from \Cref{lem:measure_approximation}, and the event that \greedy run for a specific number of steps succeeds at sampling independent sets of size $\km = \log n - C \log\log n$ from \Cref{sampling_small}, where we choose $C$ to be an arbitrary constant satisfying $C> 20$.

We will construct a coupling between $\alg$ and $\hca$ where we let $a=\log\log n$; this will be enough because $\dtv{\hc}{\hca} = O(2^{-a})$.
Observe that by \Cref{W2_convexity}, it is enough to construct a good coupling between $\alg$ and $\hck$ for each $k$ satisfying $|k - \ks| \leq a$.
From now on, fix one such $k$.
To couple $\alg$ and $\hck$, we will sample a $\km$-independent set from $\alg$, and then choose one of the $k$-independent sets which extend it uniformly at random.
Since $k-\km = o(k)$, the expected distance between independent sets in this coupling will be small.
It remains to verify that this induces $\hck$ (or something close to it) on $k$-independent sets.

In order to understand the relationship between the $\km$ and $k$ size independent sets, consider the bipartite graph $\bH = (L,R,E)$ where $L$ consists of the independent sets of size $\km$, $R$ consists of the independent sets of size $k$, and $(S,T)\in E$ iff the $\km$ independent set $S$ is contained in the $k$-independent set $T$.
Let $P$ be the transition operator of the simple random walk on $\bH$.
It is well known that one stationary measure $\mu^{\bH}$ of $P$ assigns to $v\in L\cup R$ probability proportional to its degree in $\bH$. 
Note that the degree of a $\km$-independent set $S$ in $\bH$ is precisely $\updeg_\ell(S)$.
Let $\mu_L^{\bH},\mu_R^{\bH}$ be the measure conditioned on being in $L$ or $R$, respectively.

We apply \Cref{updeg_concentration} and a union bound to argue that the degree of every $\km$-independent set in $L$ will be close to its expectation.
Indeed, observe that if we let $\km = \log n - C \log\log n$, then $\ell \defeq k-\km \in [(C-2)\log \log n, C\log \log n]$, and thus $\ell$ satisfies the conditions of \Cref{updeg_concentration}.
We can then conclude that for any set $S\subseteq[n]$ such that $|S|=\km$ and $\eps \in (0,1)$, 
\begin{align*}
    \Pr[|\updeg_\ell(S) - \E \updeg_\ell(S)| \geq \epsilon \E \updeg_\ell(S)] \leq \left(\frac{1}{\eps \log n}\right)^{\Omega(\log^5 n /\log\log n)}.
\end{align*}
Choosing $\eps = \frac{1}{\log\log n}$ and taking a union bound over all $\km$-sets in $[n]$, 
\begin{align*}
    \Pr\left[\exists S \in \binom{[n]}{\km}\, s.t. \, |\updeg_\ell(S) - \E \updeg_\ell(S)| \geq  \frac{\E \updeg_\ell(S)}{\log\log n} \right] 
&\leq \binom{n}{\km} \left(\frac{\log\log n}{\log n}\right)^{\Omega(\log^5 n /\log\log n)}\\
&\leq n^{\log n} \left(\frac{\log\log n}{\log n}\right)^{\Omega(\log^5 n /\log\log n)}\\
&\le \exp(- \Omega(\log^5 n)).
\end{align*}

Conditioning on the event that the up-degrees of all $\km$-sets concentrate, we have that the degree of every $\km$-independent set $S \in L$ is within a $(1\pm \eps)$ factor of its expectation, with $\eps = 1/\log\log n$.

We now bound
\begin{align*}
\dtv{\mu_L^{\bH}}{\mu_{\text{hc},\km}} = \sum_{S\text{ i.s. in } G, |S|=\km} \left| \frac{1}{\bZ_{\km}} - \mu_L^{\bH}(S) \right|.
\end{align*}
Letting $D$ be $D=\E \updeg_\ell(S)$ (by symmetry the same for all $S$),
\begin{align*}
\mu_L^{\bH}(S) = \frac{\updeg_\ell(S)}{\sum_{S'\in L}\updeg_\ell(S')}\in \frac{(1\pm\epsilon)D}{\bZ_{\km}(1\pm \epsilon)D} \subseteq (1\pm O(\epsilon))\frac{1}{\bZ_{\km}}
\end{align*}
so 
\begin{align*}
\dtv{\mu_L}{\mu_{\text{hc},\km}} = O(\epsilon) = o(1).
\end{align*}
Recall that $\mu_L^{\bH} P = \mu_R^{\bH} = \hck$, where the last equality is because the $\bH$-degree of every $T \in R$ is the same, $\binom{k}{\km}$.
Putting this together,
\begin{align*}
\dtv{\alg P}{\hck} = \dtv{\alg P}{\mu_L^{\bH} P} \leq \dtv{\alg}{\mu_L^{\bH}} \leq \dtv{\alg}{\mu_{\text{hc},\km}} + \dtv{\mu_{\text{hc},\km}}{\mu_L^{\bH}} \leq o(1)
\end{align*}
where the first inequality follows from the data processing inequality, and we invoke the success of \greedy (\Cref{sampling_small}) in the last inequality.

By \Cref{coupling}, there exists a coupling $(\bX',\bY)$ of $\alg P,\hck$ that agrees with $1-o(1)$ probability.
It remains to find a coupling $\pi$ of $\alg$ and $\hck$ that is close in $\w_2$.
To do this, first sample from the coupling of $\alg P,\hck$ to get $(\bX',\bY)$.
Let $P'$ be the time-reversal of $P$ with respect to $\alg$.
By \Cref{timereversal}, we can apply $P'$ to $\bX'$ to get a $\km$-independent set distributed as $\alg$, which we call $\bX$.
Again by \Cref{timereversal}, since $P'$ only removes $O(\log\log n)$ nodes,
\begin{align*}
\E_{(\bX,\bY)\sim\pi}\left\|\ \frac{\bX}{m_2(\alg)}-\frac{\bY}{m_2(\hca)}\right\|^2  &= 2-2\frac{\E_{\pi}\langle \bX,\bY\rangle}{m_2(\alg) m_2(\hca)} \\
&\leq 2 - 2\frac{(1-o(1))(\log n - O(\log \log n))}{(1-o(1))\log n} = o(1)
\end{align*}
because as long as $\bX'=\bY$ in the coupling of $\alg P$ and $\hck$, after applying $P'$ to $\bX'$ to get $\bX$, $\langle \bX,\bY\rangle$ will be $\log n - O(\log \log n)$. Since the above argument is true for any $k\in\ks \pm a$, by \Cref{W2_convexity}, $\w_2(\alg, \hca)^2 = o(1)$.
We will now show $\w_2(\hca,\hc)^2=o(1)$, and we will be done by the triangle inequality.
First, observe that 
\begin{align*}
m_2(\hc)^2 \in \log n \pm O(\log \log n) + O\left(2^{-a}\log n\right) = (1\pm o(1))\log n.
\end{align*}
because $\dtv{\hc}{\hca} = O(2^{-a})$ and the largest independent set in $\bG$ has size $O(\log n)$ with probability $1-o(1)$ by \Cref{largest_clique}. To conclude, using the coupling between $\hc,\hca$ from the $O(2^{-a})$ TV distance, we see that
\begin{align*}
\w_2(\hc,\hca)^2 \leq 2 - 2\frac{(1-O(2^{-a}))(\log n-O(\log \log n))}{(1\pm o(1))\log n} = o(1). \qquad \qedhere
\end{align*}
\end{proof}

\subsection{Auxiliary Lemmas}

In this section, we prove some combinatorial lemmas used in our concentration arguments above.

\begin{lemma}
\label{vertex_overlap}
For any set of $d$ size-$\ell$ subsets $A_1,\ldots, A_d \in \binom{[m]}{\ell}$ with a large enough union $|A_1\cup \cdots \cup A_d| > d\ell -d/2$, there exists $j\in [d]$ such that $A_j$ is disjoint from the rest of the $A_i$s.
\end{lemma}

\begin{proof}
We prove the contrapositive.
Consider a graph $H$ constructed by sequentially adding vertices $1\ldots,d$: when vertex $i$ is added, we add the edge $(i,j)$ for the smallest $j < i$ such that $A_i \cap A_j \neq \emptyset$ (if no such $j$ exists we add no edge).
To each such edge $(i,j)$, we can assign at least one element of $A_i$ which is a duplicate of an element already appearing in $A_1 \cup \cdots \cup A_{i-1}$.
Hence $|A_1 \cup \cdots \cup A_j| \le d - |E(H)|$.
If no $A_i$ is disjoint from all the rest, then the minimum degree in $H$ is at least $1$, and since $2|E(H)| =\sum_{i=1}^d \deg_H(i) \ge d$, we have our conclusion. 
\end{proof}

If $A$ is a set, then we use the notation $\binom{A}{2} = \{\{a,b\} \mid a,b \in A,\, a \neq b\}$.
\begin{lemma}
\label{edge_bound}
Let $A_1,\dots, A_d \in \binom{[m]}{\ell}$, and suppose that $|A_1\cup \cdots \cup A_d| =  d\ell-a$, and that $d\ell <  m$.
Then 
\begin{align*}
\left|\binom{A_1}{2}\cup\cdots\cup\binom{A_d}{2}\right| \geq d\binom{\ell}{2} - \min\left\{\frac{a}{\ell}\binom{\ell}{2}, \binom{a}{2}\right\}.
\end{align*}
\end{lemma}

\begin{proof}
Let $S=\emptyset$, and consider adding the sets $A_i$ one at a time to $S$, starting with $A_1$, and ending with $A_d$.
Let $B_i = A_i \setminus (A_1 \cup \cdots \cup A_{i-1})$ be the new vertices added to $S$ in step $i$, and call $|B_i| = \ell-a_i$.
By definition,
\begin{align*}
    |S|= \sum_{i=1}^d |B_i| = \sum_{i=1}^d \ell-a_i = d\ell-\sum_{i=1}^da_i,
\end{align*}
so $a= \sum a_i$.

We will now track the number of pairs added to $\bigcup_{i=1}^d \binom{A_i}{2}$ during this process.
At step $i$, we add all pairs $\{a,b\} \in \binom{A_i}{2}$, excluding those that already appear in $\binom{A_i \setminus B_i}{2}$.
Thus, the number of pairs added in step $i$ is $\binom{\ell}{2} - \binom{a_i}{2}$.
The result now follows because 
\begin{align*}
\sum_{i=1}^d \binom{a_i}{2} \leq \min\left\{\frac{a}{\ell}\binom{\ell}{2}, \binom{a}{2}\right\}
\end{align*}
where the last inequality follows because $a=\sum a_i, \binom{x}{2} + \binom{y}{2} \leq \binom{x+y}{2}$, and $a_i\leq \ell$ for all $i\in[d]$.
\end{proof}

\begin{lemma}
\label{lem:config_bound}
Let $N_a$ be the number of ways to choose $A_1,\dots, A_d \in \binom{[m]}{\ell}$ such that $|A_1\cup \cdots \cup A_d| =  d\ell-a$. Then
\begin{align*}
N_a \leq \binom{m}{d\ell-a} \binom{d\ell-1}{a-1}\frac{(d\ell)!}{(\ell!)^d}.
\end{align*}
\end{lemma}
\begin{proof}
First we choose the elements in $T = A_1\cup \cdots \cup A_d$, for which there are $\binom{m}{d\ell-a}$ choices.
Now we must choose the multiplicity with which each element in $T$ appear in the multiset-union of the $A_i$; each element must have multiplicity at least $1$, and then there are $\sum_{i=1}^d |A_i| - |T| = a$ excess appearances that we must partition among the $d\ell-a$ elements of $T$.
By a stars-and-bars argument, the number of ways to distribute these excess appearances is at most $\binom{d\ell-1}{a-1}$.
Finally, letting $T'$ be the multiset resulting from $T$ when each element is duplicated according to its multiplicity, the number of ways to form an ordered partition of $T'$ into $d$ subsets of size $\ell$ is at most $(d\ell)!/(\ell!)^d$.
The product of these bounds gives our upper bound.
\end{proof}

\ifnum\anon=0
\section*{Acknowledgments}

We wish to thank Nima Anari and Sidhanth Mohanty for useful pointers to the sampling literature.

\else
\fi

\bibliographystyle{alpha}
% \bibliography{reference}

\begin{thebibliography}{LMRW24}

\bibitem[ALS22]{ALS22}
Emmanuel Abbe, Shuangping Li, and Allan Sly.
\newblock Binary perceptron: efficient algorithms can find solutions in a rare well-connected cluster.
\newblock In {\em Proceedings of the 54th Annual {ACM} {SIGACT} Symposium on Theory of Computing}, pages 860--873, 2022.

\bibitem[AS16]{alon2016probabilistic}
Noga Alon and Joel~H Spencer.
\newblock {\em The probabilistic method}.
\newblock John Wiley \& Sons, 2016.

\bibitem[BAJ18]{BJ18}
G{\'e}rard Ben~Arous and Aukosh Jagannath.
\newblock Spectral gap estimates in mean field spin glasses.
\newblock {\em Communications in Mathematical Physics}, 361:1--52, 2018.

\bibitem[BGG{\v{S}}24]{BGGS24}
Ivona Bez{\'a}kov{\'a}, Andreas Galanis, Leslie~Ann Goldberg, and Daniel {\v{S}}tefankovi{\v{c}}.
\newblock Fast sampling via spectral independence beyond bounded-degree graphs.
\newblock {\em {ACM} Transactions on Algorithms}, 20(1):1--26, 2024.

\bibitem[BS20]{BS20}
Nikhil Bansal and Joel~H Spencer.
\newblock On-line balancing of random inputs.
\newblock {\em Random Structures \& Algorithms}, 57(4):879--891, 2020.

\bibitem[BST16]{BST16}
Nayantara Bhatnagar, Allan Sly, and Prasad Tetali.
\newblock Decay of correlations for the hardcore model on the d-regular random graph.
\newblock {\em Electronic Journal of Probability}, 21, 2016.

\bibitem[CCC{\etalchar{+}}25]{CCCYZ25}
Xiaoyu Chen, Zejia Chen, Zongchen Chen, Yitong Yin, and Xinyuan Zhang.
\newblock Rapid mixing on random regular graphs beyond uniqueness.
\newblock {\em arXiv preprint arXiv:2504.03406}, 2025.

\bibitem[Cel24]{Cel24}
Michael Celentano.
\newblock Sudakov--{F}ernique post-{AMP}, and a new proof of the local convexity of the {TAP} free energy.
\newblock {\em The Annals of Probability}, 52(3):923--954, 2024.

\bibitem[EAG24]{AG24}
Ahmed El~Alaoui and David Gamarnik.
\newblock Hardness of sampling solutions from the symmetric binary perceptron.
\newblock {\em arXiv preprint arXiv:2407.16627}, 2024.

\bibitem[EAMS22]{AMS22}
Ahmed El~Alaoui, Andrea Montanari, and Mark Sellke.
\newblock Sampling from the {S}herrington-{K}irkpatrick {G}ibbs measure via algorithmic stochastic localization.
\newblock In {\em 2022 {IEEE} 63rd Annual Symposium on Foundations of Computer Science ({FOCS})}, pages 323--334. {IEEE}, 2022.

\bibitem[EAMS23]{AMS23}
Ahmed El~Alaoui, Andrea Montanari, and Mark Sellke.
\newblock Sampling from mean-field {G}ibbs measures via diffusion processes.
\newblock {\em arXiv preprint arXiv:2310.08912}, 2023.

\bibitem[EAMS25]{AMS25}
Ahmed El~Alaoui, Andrea Montanari, and Mark Sellke.
\newblock Shattering in pure spherical spin glasses.
\newblock {\em Communications in Mathematical Physics}, 406(5):1--36, 2025.

\bibitem[EF23]{EF23}
Charilaos Efthymiou and Weiming Feng.
\newblock On the mixing time of {G}lauber dynamics for the hard-core and related models on $g(n, d/n)$.
\newblock In {\em 50th International Colloquium on Automata, Languages, and Programming ({ICALP} 2023)}, pages 54--1. Schloss Dagstuhl--Leibniz-Zentrum f{\"u}r Informatik, 2023.

\bibitem[Gam21]{Gamarnik21}
David Gamarnik.
\newblock The overlap gap property: A topological barrier to optimizing over random structures.
\newblock {\em Proceedings of the National Academy of Sciences}, 118(41):e2108492118, 2021.

\bibitem[HMP24]{HMP24}
Brice Huang, Andrea Montanari, and Huy~Tuan Pham.
\newblock Sampling from spherical spin glasses in total variation via algorithmic stochastic localization.
\newblock {\em arXiv preprint arXiv:2404.15651}, 2024.

\bibitem[HPP25]{HPP25}
Neng Huang, Will Perkins, and Aaron Potechin.
\newblock Hardness of sampling for the anti-ferromagnetic {I}sing model on random graphs.
\newblock In {\em 16th Innovations in Theoretical Computer Science Conference (ITCS 2025)}, pages 61--1. Schloss Dagstuhl--Leibniz-Zentrum f{\"u}r Informatik, 2025.

\bibitem[LMRW24]{LMRW24}
Kuikui Liu, Sidhanth Mohanty, Amit Rajaraman, and David~X Wu.
\newblock Fast mixing in sparse random ising models.
\newblock In {\em 2024 {IEEE} 65th Annual Symposium on Foundations of Computer Science ({FOCS})}, pages 120--128. {IEEE}, 2024.

\bibitem[LP17]{LP17}
David~A Levin and Yuval Peres.
\newblock {\em Markov chains and mixing times}, volume 107.
\newblock American Mathematical Soc., 2017.

\bibitem[LS24]{LS24}
Shuangping Li and Tselil Schramm.
\newblock Some easy optimization problems have the overlap-gap property.
\newblock {\em arXiv preprint arXiv:2411.01836}, 2024.

\bibitem[Mon21]{Mon21}
Andrea Montanari.
\newblock Optimization of the {S}herrington--{K}irkpatrick {H}amiltonian.
\newblock {\em {SIAM} Journal on Computing}, (0):{FOCS}19--1, 2021.

\bibitem[Pen22]{penev2022combinatorics}
Irena Penev.
\newblock Combinatorics and graph theory 1 \& 2.
\newblock 2022.

\bibitem[RSV{\etalchar{+}}14]{RSVVY14}
Ricardo Restrepo, Daniel Stefankovic, Juan~C Vera, Eric Vigoda, and Linji Yang.
\newblock Phase transition for {G}lauber dynamics for independent sets on regular trees.
\newblock {\em SIAM Journal on Discrete Mathematics}, 28(2):835--861, 2014.

\bibitem[Sub21]{Subag21}
Eliran Subag.
\newblock Following the ground states of full-{RSB} spherical spin glasses.
\newblock {\em Communications on Pure and Applied Mathematics}, 74(5):1021--1044, 2021.

\end{thebibliography}
\newcommand{\etalchar}[1]{$^{#1}$}

\end{document}